\documentclass{article}

\usepackage[utf8]{inputenc}
\usepackage{authblk}
\usepackage{graphicx}
\usepackage{amsmath}
\usepackage{amssymb}
\usepackage{amsthm}
\usepackage{color}
\usepackage{url}
\usepackage{hyperref}
\usepackage{enumitem}
\usepackage{stmaryrd}
\usepackage{tikz}
\usepackage{setspace}
\usepackage{verbatim}
\usepackage{float}

\usepackage{caption}
\usepackage{subcaption}

\usepackage{algpseudocode}
\usepackage{algorithm}

\newcommand{\be}{\begin{equation}}
\newcommand{\ee}{\end{equation}}
\newcommand{\ba}{\begin{array}}
\newcommand{\ea}{\end{array}}
\newcommand{\bea}{\begin{eqnarray}}
\newcommand{\eea}{\end{eqnarray}}

\newcommand{\ra}{\rangle}

\newcommand{\calM}{{\cal M }}
\newcommand{\calU}{{\cal U }}
\newcommand{\calZ}{{\cal Z }}
\newcommand{\FF}{\mathbb{F}}

\newcommand{\ZZ}{\mathbb{Z}}

\newcommand{\rs}[1]{\mathsf{rs}{(#1)}}
\newcommand{\cs}[1]{\mathsf{cs}{(#1)}}
\renewcommand{\ker}[1]{\mathsf{ker}{(#1)}}
\newcommand{\rk}[1]{\mathsf{rk}{(#1)}}

\newcommand{\ord}[1]{\mathrm{ord}(#1)}

\newcommand{\qc}{\mathsf{QC}(A,B)}
\newcommand{\cnot}[2]{\mathsf{CNOT}\quad {#1} \quad {#2}}
\newcommand{\initX}[1]{\mathsf{InitX} \quad {#1}}
\newcommand{\initZ}[1]{\mathsf{InitZ} \quad {#1}}
\newcommand{\measX}[1]{\mathsf{MeasX} \quad {#1}}
\newcommand{\measZ}[1]{\mathsf{MeasZ} \quad {#1}}
\newcommand{\idle}[1]{\mathsf{Idle} \quad {#1}}

\newtheorem{dfn}{Definition}

\newtheorem{lemma}{Lemma}

\newcommand{\Gate}[1]{\mathsf{#1}}

\newcommand{\cnotgate}{\Gate{CNOT}}

\newcommand{\dcirc}{d_{\mathsf{circ}}}

\renewcommand{\sec}[1]{\hyperref[sec:#1]{Section~\ref*{sec:#1}}}
\newcommand{\ssec}[1]{\hyperref[ssec:#1]{Subsection~\ref*{ssec:#1}}}
\newcommand{\fig}[1]{\hyperref[fig:#1]{Figure~\ref*{fig:#1}}}
\newcommand{\tab}[1]{\hyperref[table:#1]{Table~\ref*{table:#1}}}
\newcommand{\lem}[1]{\hyperref[lem:#1]{Lemma~\ref*{lem:#1}}}
\newcommand{\propos}[1]{\hyperref[propos:#1]{Proposition~\ref*{propos:#1}}}
\newcommand{\thm}[1]{\hyperref[thm:#1]{Theorem~\ref*{thm:#1}}}
\newcommand{\alg}[1]{\hyperref[alg:#1]{Algorithm~\ref*{alg:#1}}}

\newcommand{\edit}[1]{{\color{black}{#1}}}

\title{High-threshold and low-overhead fault-tolerant quantum memory}

 \author[1]{Sergey~Bravyi}
 \author[1]{Andrew~W.~Cross}
 \author[1]{Jay~M.~Gambetta}
  \author[1]{Dmitri~Maslov}
 \author[2]{Patrick~Rall}
 \author[1]{Theodore~J.~Yoder} 
 \affil[1]{\large{IBM Quantum, IBM T.J. Watson Research Center, Yorktown Heights, NY 10598 (USA)}}
 \affil[2]{IBM Quantum, MIT-IBM Watson AI Lab, Cambridge, MA 02142 (USA)}

\begin{document}
\maketitle

\begin{abstract}
Quantum error correction becomes a practical possibility
only if the physical error rate is below a threshold value that depends 
on a particular quantum code, syndrome measurement circuit, and decoding algorithm.
Here we present an end-to-end quantum error correction
protocol that implements fault-tolerant memory based on a family of LDPC codes with a high encoding rate
that achieves an error threshold of $0.8\%$ for the standard
circuit-based noise model. This is on par with the surface code which has remained an uncontested leader
in terms of its high error threshold for nearly 20 years.
The full syndrome measurement cycle for a length-$n$ code
in our family 
requires $n$ ancillary qubits
and a depth-7 circuit composed of nearest-neighbor CNOT gates. 
The required qubit connectivity is a degree-6 graph
that consists of two edge-disjoint planar subgraphs. 
As a concrete example, we
show that 12 logical qubits can be preserved for \edit{nearly one million} syndrome cycles
using 288 physical qubits in total, assuming the physical error rate of $0.1\%$.
We argue that achieving the same level of error suppression on 12 logical qubits with the surface code would require
\edit{nearly  3000 physical qubits}.
Our findings bring demonstrations of
a low-overhead fault-tolerant quantum memory within the reach of near-term quantum processors. 
\end{abstract}

\section{Introduction}

Quantum computing attracted attention due to its ability to offer asymptotically faster solutions to a set of computational problems compared to the best known classical algorithms \cite{nielsen2002quantum}.  It is believed that a scalable functioning quantum computer may help solve computational problems in such areas as scientific discovery, materials research, chemistry, and drug design, to name a few \cite{lloyd1996universal,wang2008quantum,reiher2017elucidating,alexeev2021quantum}.  

The main obstacle to building a quantum computer is the fragility of quantum information, owing to various sources of noise affecting it.  Since isolating a quantum computer from external effects and controlling it to induce a desired computation are in conflict with each other, noise appears to be inevitable.  The sources of noise include imperfections in qubits, materials used, controlling apparatus, State Preparation and Measurement (SPAM) errors, and a variety of external factors ranging from local man-made, such as stray electromagnetic fields, to those inherent to the Universe, such as cosmic rays. See Ref. \cite{gambetta2017building} for a summary.  While some sources of noise can be eliminated with better control \cite{mundada2023experimental}, materials \cite{de2021materials}, and shielding \cite{10.1063/1.3658630, vepsalainen2020impact, PRXQuantum.4.020356}, a number of other sources appear to be difficult if at all possible to remove.  The latter kind can include spontaneous and stimulated emission in trapped ions \cite{wu2018noise,boguslawski2023raman}, and the interaction with the bath (Purcell Effect) \cite{Houck2008} in superconducting circuits---covering both leading quantum technologies.  
Thus, error correction becomes a key requirement for building a functioning scalable quantum computer.

The possibility of quantum fault tolerance was established earlier~\cite{shor1995scheme}.  Encoding a logical qubit redundantly into many physical qubits enables one to diagnose and correct errors by repeatedly measuring syndromes of parity check operators.  However, error correction is only beneficial if the hardware  error rate is below a certain threshold value that depends on a particular error correction protocol.  The first proposals for quantum error correction, such as concatenated codes~\cite{aharonov1997fault,kitaev1997quantum,aliferis2005quantum}, focused on demonstrating the theoretical possibility of error suppression. As understanding of quantum error correction and the capabilities of quantum technologies matured, the focus shifted to finding practical quantum error correction protocols. This resulted in the development of the surface code \cite{2003faulkitaevt,bravyi1998quantum,dennis2002topological,fowler2009high} that offers a high error threshold close to $1\%$, fast decoding algorithms, and compatibility with the existing quantum processors relying on 2-dimensional (2D) square lattice qubit connectivity.  Small examples of the surface code with a single logical qubit have been already demonstrated experimentally by several groups~\cite{takita2016demonstration,marques2022logical,krinner2022realizing,zhao2022realization,google2023suppressing}.  However, scaling up the surface code to a hundred or more logical qubits would be prohibitively expensive due to its poor encoding efficiency.  This spurred interest in more general quantum codes known as Low-Density Parity-Check (LDPC) codes~\cite{gottesman2013fault}.  Recent progress in the study of LDPC codes suggests that they can achieve quantum fault-tolerance with a much higher encoding efficiency~\cite{tremblay2022constant}. Here, we focus on the study of LDPC codes, as our goal is to find quantum error correction codes and protocols that are both efficient and possible to demonstrate in practice, given the limitations of quantum computing technologies. 

A quantum error correcting code is of LDPC type if each check operator of the code acts only on a few qubits and each qubit participates only in a few checks.   Multiple variants of the LDPC codes have been proposed recently including hyperbolic surface codes~\cite{breuckmann2016constructions,higgott2021subsystem,higgott2023constructions}, hypergraph product~\cite{tillich2013quantum}, balanced product codes~\cite{breuckmann2021balanced}, two-block codes based on finite groups~\cite{kovalev2013quantum,panteleev2021degenerate,lin2023quantum,wang2023abelian}, and quantum Tanner codes~\cite{panteleev2022asymptotically,leverrier2022quantum}.  The latter were shown~\cite{panteleev2022asymptotically,leverrier2022quantum} to be asymptotically ``good'' in the sense of offering a constant encoding rate and linear distance -- a parameter quantifying the number of correctable errors. In contrast, the surface code has an asymptotically zero encoding rate and only square-root distance. Replacing the surface code with a high-rate, high-distance LDPC code could have major practical implications. First, fault-tolerance overhead (the ratio between the number of physical and logical qubits) could be reduced dramatically. Secondly, high-distance codes exhibit a very sharp decrease in the logical error rate: as the physical error probability crosses the threshold value, the amount of error suppression achieved by the code can increase by orders of magnitude even with a small reduction of the physical error rate. This feature makes high-distance LDPC codes attractive for near-term demonstrations which are likely to operate in the near-threshold regime.  
\edit{However, it was previously believed that outperforming the surface code for realistic noise models 
 including memory, gate, and SPAM errors may require very large LDPC codes with more than 10,000 physical qubits~\cite{higgott2021subsystem}.}

Here we present several concrete examples of high-rate LDPC codes with a few hundred physical qubits equipped with a low-depth syndrome measurement circuit, an efficient decoding algorithm, and a fault-tolerant protocol for addressing individual logical qubits.  These codes exhibit an error threshold close to 1\%, show excellent performance in the near-threshold regime, and offer \edit{more than 10X reduction} of the encoding overhead compared with the surface code.  Hardware requirements for realizing our error correction protocols are relatively mild, as each physical qubit is coupled by two-qubit gates with only six other qubits. Although the qubit connectivity graph is not locally embeddable into a 2D grid, it can be decomposed into two planar degree-3 subgraphs.  As we argue below, such qubit connectivity is well-suited for architectures based on superconducting qubits.  Before stating our results, let us describe several must-have  features for a quantum error-correcting code to be suitable for near-term experimental demonstrations
and formally pose the problem addressed in this work.


\section{Code selection criteria}

In this work, we study the problem of realizing a fault-tolerant quantum memory with a small qubit overhead and a large code distance.  Our goal is to construct a combination of the LDPC code, syndrome measurement circuitry, and the decoding (error correction) algorithms, suitable for a near-term demonstration, but also offering long-term utility, while taking into account the capabilities and limitations of the superconducting circuits quantum hardware.  In other words, we seek to develop a practical error correction protocol.  Our selection criteria reflect this goal.

We focus on encoding $k\gg 1$ logical qubits into  
$n$ data qubits and use $c$ ancillary check qubits to measure the error syndrome.
In total, the code relies on $n\,{+}\,c$ physical qubits. The net encoding rate 
 is therefore
\[
r=\frac{k}{n+c}.
\]
For example, the standard surface code architecture encodes $k\,{=}\,1$ logical qubit into $n\,{=}\,d^2$ data qubits for a distance-$d$ code and uses $c\,{=}\,n{-}1$ check qubits for syndrome measurements.
The net encoding rate is $r\approx 1/(2d^2)$, which quickly becomes impractical as one is forced to choose a large code distance, due to, for instance, the physical errors being close to the threshold value.  In contrast, we seek a high-rate LDPC code with $r\gg 1/d^2$.

To prevent the accumulation of errors one must be able to  measure the 
error syndrome   frequently enough. This is accomplished by a syndrome measurement (SM) circuit that couples data qubits in the support of each check operator with the respective ancillary qubit by a sequence of $\cnotgate$ gates.
Check qubits are then measured revealing the value of the error syndrome.
The time it takes to implement the SM circuit is proportional to its depth
— the number of gate layers composed of non-overlapping $\cnotgate$s.
Since new errors continue to occur while the SM circuit is executed, its depth
should be minimized.
Thus we seek an LDPC code with a high rate $r$ and low-depth SM circuit.

A noisy version of the SM circuit may include 
several types of faulty operations such as
memory errors
on  data or check qubits, faulty $\cnotgate$ gates, qubit initializations
and measurements. 
We consider the circuit-based noise model~\cite{fowler2009high}
where each operation 
fails  with the probability $p$.
Faults on different operations are independent.
A logical error occurs when the final error-corrected state of $k$ logical qubits differs from the initial encoded state. 
The probability of a logical error $p_L$ depends on the  error rate $p$,
details of the SM circuits, and a decoding algorithm.
A pseudo-threshold $p_0$ of an error correction protocol is defined
as a solution of  the break-even equation $p_L(p)=kp$.
Here $kp$ is an estimate of the probability that at least one of $k$ unencoded qubits suffers from an error.
 To achieve a significant error suppression in the regime $p\,{\sim}\,10^{-3}$, which is relevant for near-term demonstrations, it is desirable to have pseudo-threshold 
close to 1\% or higher.
For example, the surface code architecture achieves  pseudo-threshold $p_0\approx 1\%$ for a large enough code distance~\cite{fowler2009high}.
We seek a high-rate LDPC code with a low-depth 
SM circuit and a high pseudo-threshold.

A logical error is undetectable if it can be generated without
triggering any syndromes.  Such errors 
span at least $d$ data qubits for a distance-$d$ code.
Let us say that a SM circuit has distance $\dcirc$ if
it takes at least $\dcirc$ faulty operations in the circuit to generate an undetectable logical error.
By definition, $\dcirc\,{\le}\,d$ for any distance-$d$ code
and typically $\dcirc\,{<}\,d$ since a few faulty operations in the SM circuit may create a high-weight error on the data qubits. 
We say that a SM circuit is  distance-preserving
if $\dcirc\,{=}\,d$ meaning the  circuit is designed so as to avoid accumulating high-weight errors, which is the best one can hope for. 
It is preferred (but not required)  that the SM circuit is distance-preserving.

Another criterion is dictated by the limited qubit connectivity of near-term quantum devices.
Each quantum code can be described by a Tanner graph $G$ such that 
each vertex of $G$ represents either a data qubit or a check operator. 
A check vertex  $i$
and a data vertex $j$ are connected by an edge if the $i$-th check operator
acts non-trivially  
on the $j$-th data qubit (by applying Pauli $X$ or $Z$). \fig{2Dlayout}~A) shows the Tanner graph describing a distance-3 surface code. 
To keep the SM circuit depth small, it is desirable that two-qubit gates such as $\cnotgate$ can be applied along every edge of the Tanner graph. By construction, the Tanner graph of any LDPC code
has a small degree. 
One drawback of high-rate LDPC codes is that their Tanner graphs may not be locally embeddable 
into the 2D grid~\cite{baspin2022quantifying,bravyi2010tradeoffs}. This poses a challenge
for hardware implementation with superconducting qubits
coupled by   microwave resonators.
A useful VLSI design concept is graph {\em thickness},  see~\cite{mutzel1998thickness,tremblay2022constant} for details.
A graph $G\,{=}\,(V,E)$ is said to have thickness $\theta$ if one can partition its set of edges $E$ into disjoint union of $\theta$ sets 
$E_1\sqcup E_2\sqcup\ldots \sqcup E_\theta\,{=}\,E$ such that each subgraph $(V,E_i)$ is planar.
Informally, a graph with thickness $\theta$ can be viewed as a vertical stack of $\theta$ planar graphs.
Qubit connectivity described by a planar graph (thickness $\theta\,{=}\,1$) is the simplest one from hardware
perspective since the couplers do not cross. Graphs with thickness $\theta\,{=}\,2$ might still be implementable since two planar layers of couplers and their control lines can be attached to the top and the bottom side of the chip hosting qubits, and the two sides mated (see \sec{conclusions} for a detailed discussion).
Graphs with thickness $\theta\,{\ge}\, 3$ are much harder to implement. Thus we seek a high-rate LDPC code with a low-depth 
SM circuit, high pseudo-threshold, and a low-degree Tanner graph with thickness $\theta\le 2$.

Finally, the code must perform a useful function within a larger architecture for quantum computation, the simplest of which is a quantum memory. In a quantum memory it must be possible to measure every logical qubit in at least one Pauli basis, permitting initialization and readout of individual qubits. Furthermore it should be possible to connect the code to another error correction code and facilitate Pauli product measurements between their logical qubits. This enables load-store operations that transfer quantum data out of and into the code via quantum teleportation. For the purpose of the shorter-term goal of demonstrating the code in practice, the code should also feature enough logical operations to facilitate experiments to verify correct operation.

Our code selection criteria are summarized below.
\begin{enumerate}
\item We desire a code with a large distance $d$ and
a high encoding rate $r\,{\gg}\, 1/d^2$, 
\item that is complemented by a short-depth syndrome measurement circuit, 
\item offers a pseudo-threshold close to 1\% (or higher) for the circuit-based noise model,
\item is constructed over thickness-2 or less Tanner graph,
\item and possesses fault-tolerant load-store operations as well as readout and initialization of individual qubits. 
\end{enumerate}

\section{Main results}
\label{sec:results}

Here we give concrete examples of LDPC codes equipped with syndrome
measurement circuits and efficient decoding algorithms that 
meet all above conditions.
Our examples fall into the family of  tensor product generalized bicycle codes
proposed by Kovalev and  Pryadko~\cite{kovalev2013quantum}.
\edit{We named our codes Bivariate Bicycle (BB) since they are based on bivariate polynomials, as detailed below.}
These are stabilizer codes of CSS-type~\cite{steane1996multiple,calderbank1996good} 
that can be described by a collection of few-qubit check (stabilizer) operators
composed of Pauli $X$ and $Z$.
At a high level, a BB code is similar to the two-dimensional toric code~\cite{2003faulkitaevt}.
In particular, physical qubits of a BB code can be laid out on a two-dimensional  grid
with periodic boundary conditions such that all check operators are obtained from a single pair of $X$- and $Z$-checks  by applying horizontal and vertical shifts of the grid. However, in contrast to the plaquette and vertex stabilizers
describing the toric code, check operators of a BB code are not geometrically local.
Furthermore, each check acts on six qubits rather than four qubits.  See \fig{2Dlayout}~B) an example Tanner graph of a BB code.
We give a formal definition of BB codes in \sec{codedefinition}.
The Tanner graph of any BB code has vertex degree six. 
Although this graph  may not be locally embeddable into a 2D grid, we show that it has thickness $\theta\,{=}\,2$, as desired. 
This result may be surprising since it is known that 
a general degree-$6$ graph can have thickness $\theta \,{=}\, 3$, see~\cite{mutzel1998thickness}.

\begin{figure}[ht]
  \centering
    \includegraphics[height=9cm]{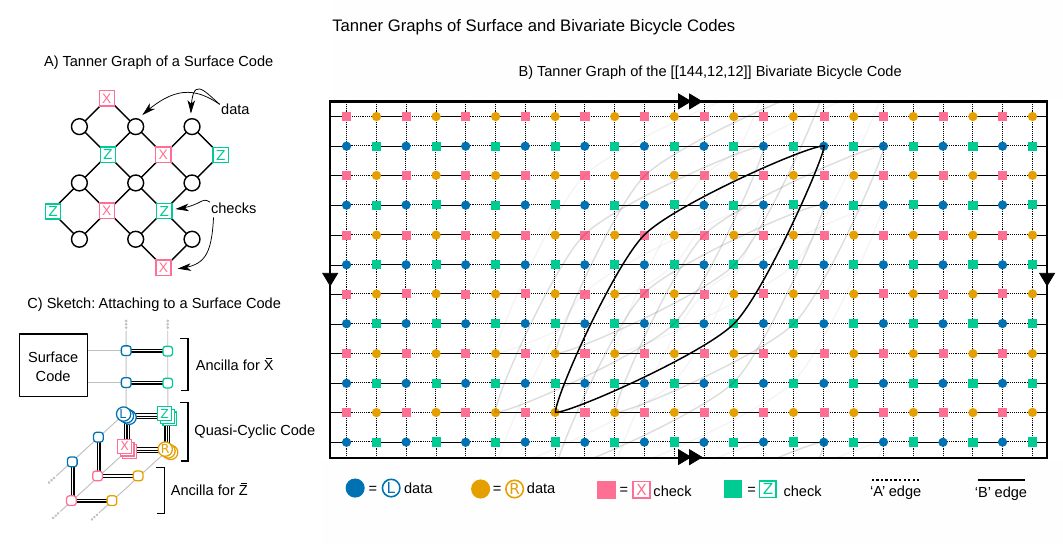}
         \caption{A) Tanner graph of a surface code, for comparison. 
         C) Tanner graph of a Bivariate Bicycle code with parameters $[[144,12,12]]$  embedded into a torus.
     Any edge of the Tanner graph connects a data and a check vertex.
      Data qubits associated with the registers $q(L)$ and $q(R)$ are shown by bLue and oRange circles. 
     Each vertex has six incident edges including four short-range edges (pointing north, south, east, and west) and two long-range edges. There are also several long-range edges, of which we only show a few to to avoid clutter.  Dashed and solid edges indicate two planar subgraphs spanning the Tanner graph, see \sec{codedefinition}.     
      B) Sketch of a Tanner graph extension for measuring $\bar Z$ and $\bar X$ following \cite{cohen2022lowoverhead}.
         The ancilla corresponding to the $\bar X$ measurement can be connected to a surface code, enabling load-store operations for all logical qubits via quantum teleportation and some logical unitaries. This extended Tanner graph has a thickness-2 implementation, see \sec{logicals}.
         }
  \label{fig:2Dlayout}
\end{figure}

Below we use the standard notation $[[n,k,d]]$ for code parameters.
Here $n$ is the code length (the number of data qubits), $k$ is the number of logical  qubits,
and $d$ is the code distance. \tab{codes_intro} shows small examples
of BB codes along with several metrics of the error suppression 
achieved by each codes. 
 The distance-12 code $[[144,12,12]]$ may be the most promising for near-term demonstrations, as it combines large distance and high net encoding rate $r\,{=}\,1/24$.
For comparison, the distance-13 surface code has net encoding rate $r\,{=}\,1/338$.
Below we show that the distance-12
BB code  outperforms the distance-13 surface code for the experimentally relevant range of error rates, see \fig{numerics} B).
To the best of our knowledge, all codes shown in \tab{codes_intro} are new.

\begin{table}[t]
\begin{center}
\begin{tabular}{c|c|c|c|c|c}
\hline
& & & &    \\
$[[n,k,d]]$ & \parbox{3cm}{\centering Net Encoding \\ Rate $r$}  & 
\parbox{3cm}{\centering Circuit-level\\ distance $\dcirc$} &
Pseudo-threshold $p_{0}$ & 
$p_L(0.001)$ & $p_L(0.0001)$ \\
& & & & \\
\hline
\hline
$[[72,12,6]]$ &  1/12 & $\le 6$ &  $0.0048$ & $7\times 10^{-5}$ & $7 \times 10^{-8}$    \\
\hline 
$[[90,8,10]]$ & $1/23$ & $\le 8$ &   $0.0053$  & $5\times 10^{-6}$ & $4 \times 10^{-10}$ \\
\hline
$[[108,8,10]]$ &  1/27 & $\le 8$ & $0.0058$  & 
$3\times 10^{-6}$ & $1\times 10^{-10}$ \\
\hline
$[[144,12,12]]$ & $1/24$ & $\le 10$ & $0.0065$  & $2\times 10^{-7}$ & $8\times 10^{-13}$ \\
\hline
$[[288,12,18]]$ & $1/48$ & $\le 18$ &  $0.0069$  & $2\times 10^{-12}$ & $1\times 10^{-22}$ \\
\end{tabular}
\end{center}
\caption{Small examples of Bivariate Bicycle LDPC codes and their 
performance for the circuit-based noise model.
All codes have weight-6 checks, thickness-2 Tanner graph, and a  depth-$7$ syndrome measurement circuit. A code with parameters $[[n,k,d]]$
requires $2n$ physical qubits in total
and achieves the net encoding rate $r=k/2n$
(we round $r$ down to the nearest inverse integer).
Circuit-level distance $\dcirc$
is the minimum number of faulty operations in the syndrome measurement circuit required to generate an undetectable logical error.
The pseudo-threshold $p_{0}$ is
a solution of the  break-even equation $p_L(p)=kp$,
where $p$ and $p_L$ are the physical and logical error rates
respectively. The logical error rate $p_L$ was computed numerically for $p\ge 10^{-3}$
and extrapolated to lower error rates.
}
\label{table:codes_intro}
\end{table}

To quantify the level of error suppression achieved by a code
we introduce 
SM circuits that repeatedly measure
the syndrome of each check operator. The full cycle of syndrome measurement for
a length-$n$ BB code requires $n$ ancillary check qubits
to store the measured syndromes.
According, the net encoding rate is $r=k/(2n)$.
Check qubits are coupled with the data qubits
by applying a sequence of $\cnotgate$ gates. The full cycle of syndrome measurement
requires only $7$ layers of $\cnotgate$s regardless of the code length.
The check qubits are initialized and measured at the beginning and at the end of the syndrome cycle respectively, see \sec{syndrome_circuit} for details.
We emphasize that our SM circuit applies to any BB code beyond those listed in \tab{codes_intro}.
The circuit respects the cyclic shift symmetry of the underlying code.
Assuming that the physical qubits (data or check) are located at vertices
of the Tanner graph, all $\cnotgate$ gates in the SM circuit act on nearest-neighbor
qubits. Thus the required qubit connectivity is described by a degree-6 thickness-2
graph, as desired. 
We conjecture, based on the numerical simulations, that our SM circuit is distance-preserving for the code $[[72,12,6]]$, see \tab{codes_intro} for the upper bounds on $\dcirc$ (the upper bound $\dcirc\le 18$ for the 288-qubit code is unlikely to be tight \edit{and this affects the fit and extrapolations}).

The full error correction protocol performs $N_c\gg 1$ syndrome measurement
cycles  and calls a decoder — a classical algorithm
that takes as input the measured syndromes and outputs a guess of the final
error on the data qubits. 
Error correction succeeds if the guessed and the actual error coincide modulo a product of check operators. In this case the two errors have the same action on any encoded (logical) state.
Thus applying the inverse of the guessed error would return data qubits to the initial logical sate.
Otherwise, if the guessed and the actual error differ by a non-trivial logical operator, error correction fails resulting in a logical error.
Our numerical experiments are based on the Belief Propagation with an Ordered Statistics Decoder (BP-OSD)
proposed by Panteleev and Kalachev~\cite{panteleev2021degenerate}.
The original work~\cite{panteleev2021degenerate} described BP-OSD in the context of a toy noise model with memory errors only.
Here we show how to extend BP-OSD to the circuit-based noise model.
\edit{Our approach closely follows Refs.~\cite{ delfosse2023spacetime,mcewen2023relaxing,higgott2023improved,geher2023tangling}.}
We also show that BP-OSD can be applied to other problems in quantum fault-tolerance such as estimating the distance of a quantum LDPC code, see \sec{decoder} for details. 
These tasks can be accomplished with a relatively minor extension of the publicly available BP-OSD software developed by Roffe et al.~\cite{roffe2020decoding}

Let $P_L(N_c)$ be the logical error probability after performing
$N_c$ syndrome cycles.
Define the logical error rate as $p_L=1 - (1-P_L(N_c))^{1/N_c}\approx P_L(N_c)/N_c$.
Informally, $p_L$ can be viewed
as the logical error probability per syndrome cycle. 
Following common practice, we choose $N_c=d$ for a distance-$d$ code.
\fig{numerics} A) shows the logical error rate  achieved by codes from \tab{codes_intro}.
The logical error rate was computed numerically for $p\ge 10^{-3}$ and extrapolated to lower error rates
using  a fitting formula $p_L=p^{\dcirc'/2} e^{c_0 + c_1 p + c_2p^2}$,
where $c_0,c_1,c_2$ are fitting parameters and $\dcirc'$ is an upper bound on $\dcirc$ from \tab{codes_intro}. The observed pseudo-threshold  \edit{for the $144$-qubit and $288$-qubit codes is close to $0.007$}, which is nearly the same as the error threshold of the surface code~\cite{groszkowski2009high}.
To the best of our knowledge, this provides the first example of high-rate, large-distance LDPC codes achieving the pseudo-threshold close to 1$\%$ under the circuit-based noise model.

\begin{figure}[h]
\begin{subfigure}{.5\textwidth}
  \centering
  \includegraphics[width=7cm]{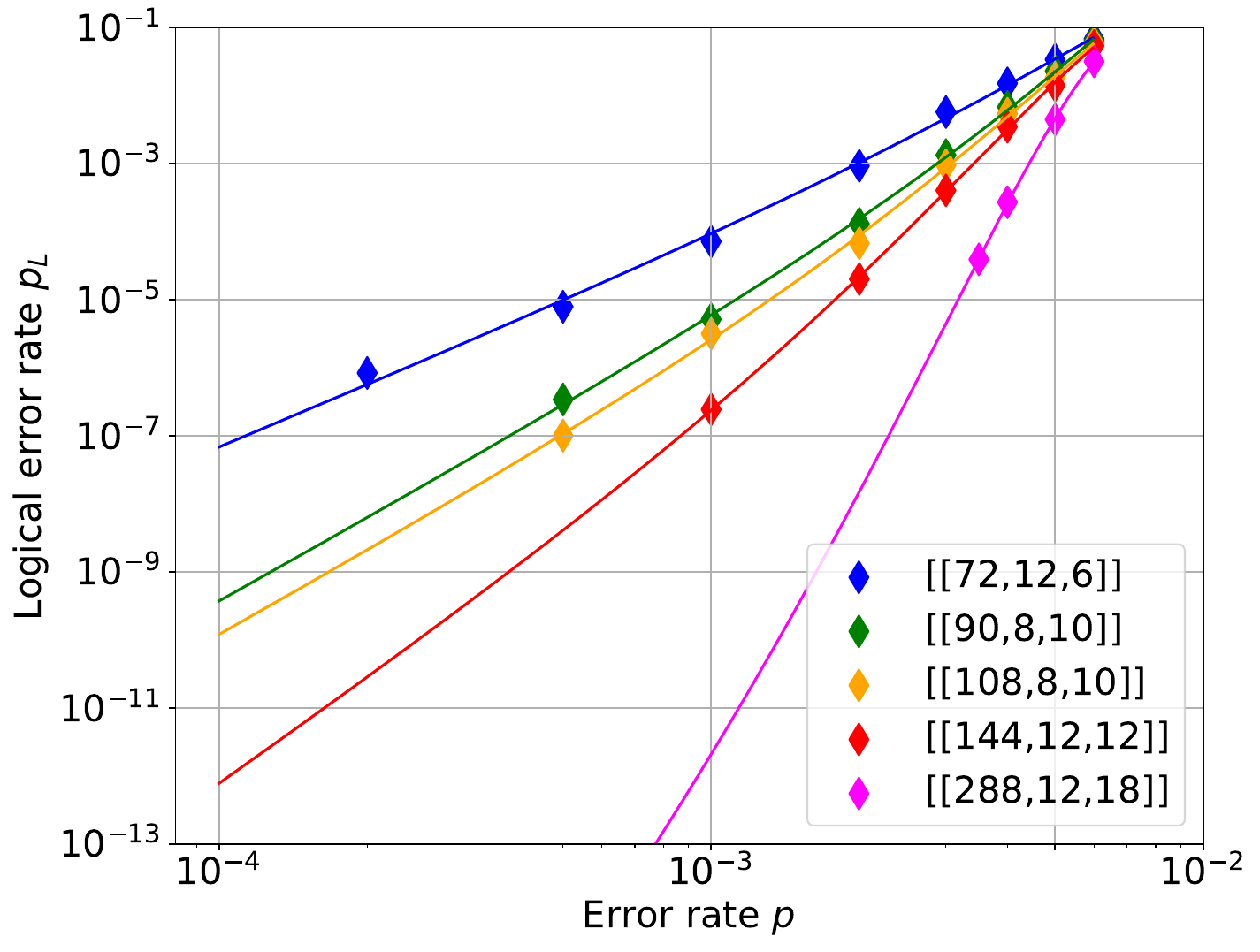}
 \caption{A subfigure}
\label{fig:sub1}
\end{subfigure}%
\begin{subfigure}{.5\textwidth}
  \centering
  \includegraphics[width=7cm]{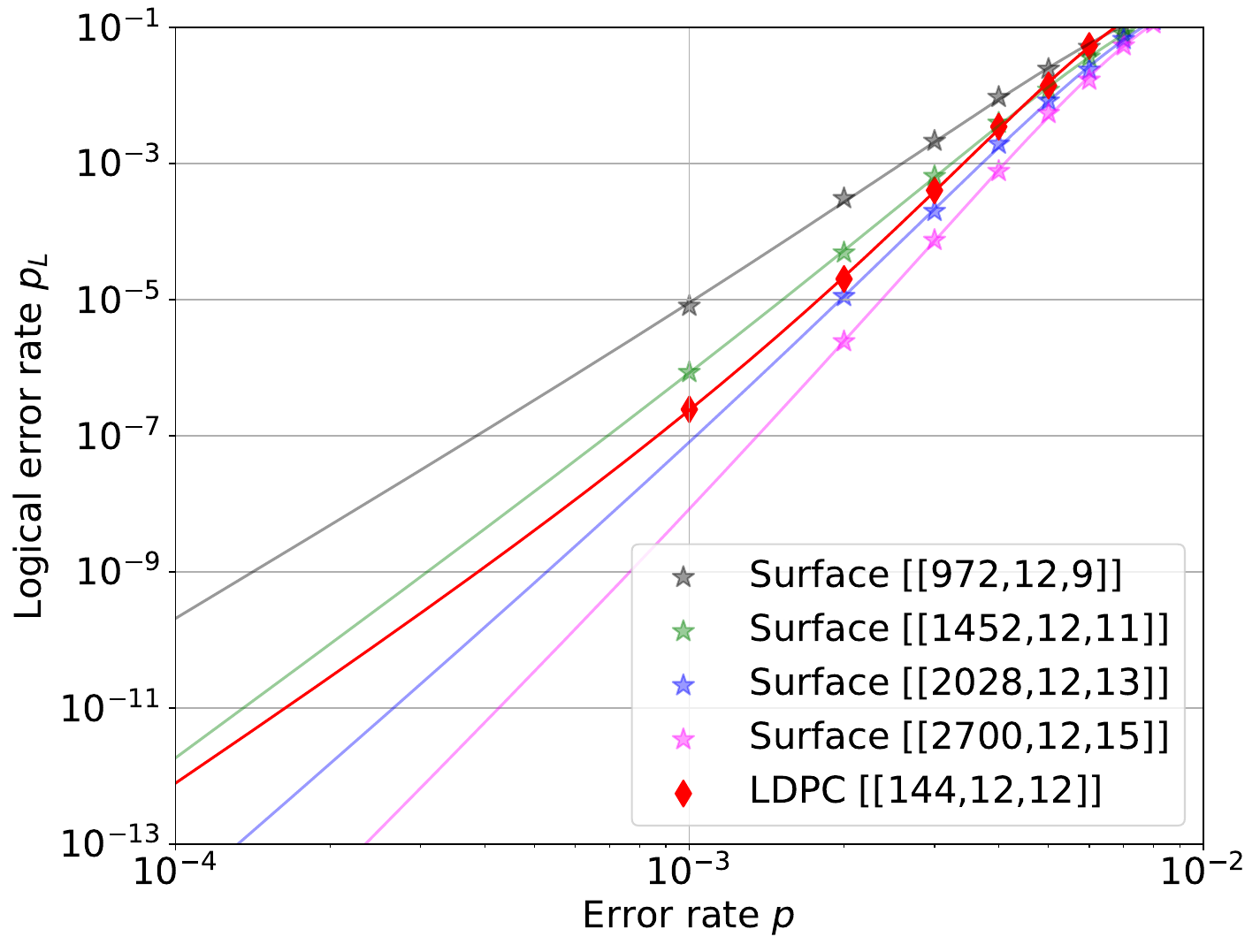}
\caption{A subfigure}
\label{fig:sub2}
\end{subfigure}
     \caption{A) Logical vs physical error rate for small examples
     of Bivariate Bicycle  LDPC codes. A numerical estimate of $p_L$ (diamonds) was
     obtained by simulating $d$ syndrome cycles for a distance-$d$ code.
     Most of the data points have error bars $\approx p_L/10$ due to sampling errors.
     B) Comparison between the Bivariate Bicycle LDPC code  $[[144,12,12]]$
     and surface codes with $12$ logical qubits and distance \edit{$d\in \{9,11,13,15\}$}.
     The distance-$d$ surface code with $12$ logical qubits has length $n=12d^2$
     since each logical qubit is encoded into a separate $d{\times}d$ patch of the surface code
     lattice.
     }
  \label{fig:numerics}
\end{figure}

For example, suppose that the physical error rate is $p=10^{-3}$, which is a realistic goal for near-term demonstrations. 
Encoding $12$ logical qubits using the distance-$12$ code from \tab{codes_intro} would offer the logical error rate \edit{$2\times 10^{-7}$} which is enough to preserve $12$ logical qubits for
\edit{nearly one million syndrome cycles}.
The total number of physical qubits required
for this encoding is  $288$.
The distance-$18$ code from \tab{codes_intro} would require $576$ physical qubits while suppressing the error rate from $10^{-3}$ to \edit{$2\times 10^{-12}$}
enabling roughly hundred billion syndrome cycles.
For comparison, encoding  $12$ logical qubits into separate patches of the surface code
would require \edit{nearly $3000$} physical qubits to suppress the error rate from $10^{-3}$ to $10^{-6}$,
see \fig{numerics}~B).
In this example the distance-$12$ BB  code offers \edit{more than $10$X} saving in the number of  physical qubits compared with the surface code.

We also find that BB LDPC codes admit extensions that allow them to function as a logical memory with load-store operations. In \sec{logicals} we show how to use methods from \cite{cohen2022lowoverhead} to attach two ancilla systems to the code that permit logical measurement of all logical qubits in the $X$ and $Z$ bases. Which logical qubit is being measured can be controlled via a set of fault tolerant unitary operations. The extended Tanner graph is not only thickness-2, but the extension from the $X$ ancilla system is ``effectively planar'' (in a sense we define later) facilitating interconnection with other codes on the same chip.

Our findings bring experimental demonstration of high-rate LDPC codes within the reach of near-term quantum processors which are expected to offer a few hundred physical qubits, gate error rates close to $10^{-3}$, and long range qubit connectivity~\cite{bravyi2022future}. 

The rest of this paper is organized as follows. \sec{codedefinition} formally defines BB LDPC codes and proves their basic properties.  The construction of the syndrome measurement circuit is detailed in \sec{syndrome_circuit}.  The circuit-based noise model and BP-OSD decoder for this noise model are discussed in \sec{decoder} with some implementation details deferred to \sec{numerics}. We describe fault tolerant memory capabilities in \sec{logicals}.  A summary of our findings and some open questions can be found in \sec{conclusions}.

\section{Bivariate Bicycle quantum LDPC codes}\label{sec:codedefinition}

Let $I_\ell$ and $S_\ell$ be the identity matrix and the cyclic shift matrix of size $\ell \times \ell$
respectively.
The $i$-th row of $S_\ell$ has a single nonzero entry equal to one at the column $i\,{+}\,1{\pmod \ell}$.
For example,
\[
S_2=\left[\ba{cc} 0 & 1 \\ 1 & 0 \\ \ea\right] \quad \mbox{and} \quad
S_3= \left[\ba{ccc} 0 & 1 & 0 \\ 0 & 0 & 1 \\ 1 & 0 & 0 \\ \ea \right].
\]
Consider matrices
\[
x=S_\ell \otimes I_m \quad \mbox{and} \quad y = I_\ell \otimes S_m.
\]
Note that $xy=yx$ and $x^\ell = y^m =I_{\ell m}$.
A BB code is defined by a pair of matrices
\be
\label{AB}
A = A_1 + A_2 + A_3  \quad \mbox{and} \quad B = B_1+B_2+B_3
\ee
where each matrix $A_i$ and $B_j$ is a power of $x$ or $y$. Here and below the addition and multiplication of binary matrices is performed modulo two, unless stated otherwise. Thus, we also assume the $A_i$ are distinct and the $B_j$ are distinct to avoid cancellation of terms. For example,
one could choose $A=x^3+y+y^2$ and $B=y^3+x+x^2$.
Note that $A$ and $B$ have exactly three non-zero entries in each row and each column. Furthermore, $AB\,{=}\,BA$ since $xy\,{=}\,yx$.
The above data defines a
BB LDPC code denoted $\qc$ with 
length $n\,{=}\,2\ell m$  and 
 check matrices
\be
\label{HXHZ}
H^X = \left[A | B\right] \quad \mbox{and} \quad  H^Z=\left[B^T | A^T\right].
\ee
Here the vertical bar indicates stacking matrices
horizontally and $T$ stands for the matrix transposition. 
Both matrices $H^X$ and $H^Z$ have size $(n/2){\times} n$.
Each row $v\,{\in}\, \FF_2^n$ of $H^X$ defines an $X$-type check operator $X(v)=\prod_{j=1}^n X_j^{v_j}$.
Each row $v\,{\in}\, \FF_2^n$ of $H^Z$ defines a $Z$-type check operator $Z(v)=\prod_{j=1}^n Z_j^{v_j}$.
Any X-check and Z-check commute since they overlap on even number of qubits (note that $H^X (H^Z)^T = AB+BA=0 {\pmod 2}$).
To describe the code parameters we use certain linear subspaces associated with the check matrices, see Table~1 for our notations.
Then the code $\qc$ has parameters $[[n,k,d]]$ with 
\be
n=2\ell m,\quad 
k = 2\cdot\mathrm{dim}\left(\ker{A} \cap \ker{B}\right)
\quad \mbox{and} \quad
d= \min\bigl\{ |v|{:} \,\, v\in \ker{H^X}{\setminus} \rs{H^Z} \bigr\},
\ee
see \lem{kd}.
Here $|v|=\sum_{i=1}^n v_i$ is the Hamming weight of a vector $v\in \FF_2^n$.
We note that the code $\qc$ can be viewed as a special case of the Lifted Product construction~\cite{panteleev2021quantum}
based on the abelian group $\ZZ_\ell \times \ZZ_m$.
Here $\ZZ_j$ denotes the cyclic group of order $j$.

\begin{table}[t]
\begin{center}
\begin{tabular}{c|c|c}
\hline
Notation & Name & Definition \\
\hline
\hline
$\rs{H}$ & row space & Linear span of rows of $H$ \\
\hline
$\cs{H}$ & column space & Linear span of columns of $H$\\
\hline
$\ker{H}$ & nullspace & Vectors orthogonal to each row of $H$\\
\hline
$\rk{H}$ & rank &  $\rk{H}=\mathrm{dim}(\rs{H})=\mathrm{dim}(\cs{H})$ \\
\hline
\end{tabular}
\caption{Notations for  linear spaces associated with a binary matrix $H$.
Here the linear span, orthogonality, and dimension are computed over the binary field
$\FF_2=\{0,1\}$.
If $H$ has size $s\,{\times}\, n$ then $\rs{H}\subseteq \FF_2^n$,
$\cs{H}\subseteq \FF_2^s$, and $\ker{H}\subseteq \FF_2^n$.
}
\end{center}
\end{table}

\begin{table}[t]
\begin{center}
\begin{tabular}{c|c|c|c|c|c}
\hline
& & & & &   \\
$[[n,k,d]]$ & \parbox{3cm}{\centering Net Encoding \\ Rate $r$}   & 
$\ell,m$ & $A$ & $B$ \\
& & & & &   \\
\hline
\hline
$[[72,12,6]]$ &  1/12 &   $6,6$ & $x^3+y+y^2$  & $y^3+x+x^2$ \\
\hline 
$[[90,8,10]]$ & $1/23$  & $15,3$ & $x^9 + y + y^2$ & $1+ x^2 +x^7$ \\ 
\hline
$[[108,8,10]]$ &  1/27 &    $9,6$ &  $x^3+y+y^2$ & $y^3+x+x^2$  \\
\hline 
$[[144,12,12]]$ & $1/24$ &   $12,6$ &  $x^3+y+y^2$ & $y^3+x+x^2$ \\
\hline
$[[288,12,18]]$ & $1/48$ &   $12,12$ & $x^3 + y^2 + y^7$  & $y^3 + x+x^2$\\
\hline
$[[360,12,\le 24]]$ & $1/60$ &  $30,6$ &  $x^9 + y + y^2$  &  $y^3 + x^{25}+x^{26}$\\
\hline
$[[756,16,\le 34]]$ & $1/95$  & $21,18$ &
$x^{3}+y^{10} + y^{17}$ & $y^5+x^3+x^{19}$\\
\end{tabular}
\end{center}
\caption{Small examples of Bivariate Bicycle LDPC codes and their parameters.  All codes have weight-6 checks, thickness-2 Tanner graph, and a depth-7 syndrome measurement circuit. 
Code distance was computed by the mixed integer programming approach of
Ref.~\cite{landahl2011fault}.
 Notation $\le d$ indicates that only an upper bound on the code distance is known at the time of this writing.  We round $r$ down to the nearest inverse integer. 
The codes have check matrices  $H^X=[A|B]$
and $H^Z=[B^T|A^T]$ with  $A$ and $B$ defined in the last two columns.
The matrices $x$, $y$ obey
$x^\ell=y^m=1$ and $xy=yx$.}
\label{table:codes}
\end{table}

\tab{codes} describes the polynomials $A$ and $B$ that give rise to examples of high-rate, high-distance  BB codes found by a numerical search. 
This includes all codes from \tab{codes_intro} and two examples of higher distance codes. 
To the best of our knowledge,
all these examples are new. 
The code $[[360,12,\le 24]]$  improves upon a code  $[[882,24,\le 24]]$ with weight-6 checks found by Panteleev and Kalachev in~\cite{panteleev2021degenerate} (assuming that our distance upper bound is  tight). Indeed, taking two independent copies of the 360-qubit code gives parameters $[[720,24,\le 24]]$.

By construction, the  code $\qc$ has 
weight-6 check operators and each qubit participates in six checks
(three $X$-type plus three $Z$-type checks). 
Accordingly, the  code $\qc$ has a degree-$6$ Tanner graph.
Below we show that the Tanner graph  has thickness $\theta\le 2$, as desired, see \lem{thickness}. 

We note that the recent work by Wang, Lin, and Pryadko~\cite{wang2023abelian,lin2023quantum}
described examples of group-based codes closely related to the codes considered here. Some of the group-based codes with weight-8 checks found in~\cite{lin2023quantum} outperform our BB codes with weight-6 checks in terms of the parameters $n,k,d$.
It remains to be seen whether group-based codes can achieve a similar or better level of error suppression for the circuit-based noise model.  

In the rest of this section we establish some
properties of BB LDPC codes.
\begin{lemma}
\label{lem:kd}
The code $\qc$ has parameters $[[n,k,d]]$, where
\[
n = 2\ell m,\quad 
k = 2\cdot\mathrm{dim}\left(\ker{A} \cap \ker{B}\right),
\quad \mbox{and} \quad
d= \min\bigl\{ |v|{:} \,\, v\in \ker{H^X}{\setminus} \rs{H^Z} \bigr\}.
\]
The code offers equal distance for $X$-type and $Z$-type errors.
\end{lemma}
\begin{proof}
It is known~\cite{steane1996multiple,calderbank1996good} that 
\[
k = n - \rk{H^X}-\rk{H^Z}.
\]
We claim that $\rk{H^X}\,{=}\,\rk{H^Z}$.
Indeed, define a self-inverse permutation matrix $C_\ell$ of size $\ell \,{\times}\, \ell$ such that the $i$-th column of $C_\ell$ has a single nonzero entry equal to one at the row $j=-i{\pmod \ell}$. Define $C_m$ similarly
and let $C=C_\ell \otimes C_m$. Since $C_\ell S_\ell C_\ell = S_\ell^T$
and $C_m S_m C_m = S_m^T$, one gets
\be
\label{ABC2}
A^T = C AC \quad \mbox{and} \quad B^T = CBC.
\ee
Therefore one can write
\[
H^Z= [B^T|A^T] = [CBC |CAC] = C[A|B]\left[ \ba{cc} 0 & C  \\ C & 0\\ \ea \right]
=CH^X\left[ \ba{cc} 0 & C  \\ C & 0\\ \ea \right].
\]
Thus $H^Z$ is obtained from $H^X$ by multiplying on the left and on the right by invertible matrices.
This implies $\rk{H^X}=\rk{H^Z}$. Therefore
\begin{align*}
k& = n-2{\cdot}\rk{H^Z}=n-2{\left(\frac{n}2-\dim{(\ker{(H^Z)^T)})}\right)} = n - 2{\left(\frac{n}2-\dim{(\ker{A}\cap \ker{B})}\right)}\\
&= 2\cdot\mathrm{dim}\left(\ker{A} \cap \ker{B}\right).
\end{align*}
Here we noted that $H^Z$ has size $(n/2)\,{\times}\, n$ and $\ker{(H^Z)^T)}=\ker{A}\cap \ker{B}$ since $H^Z=[B^T|A^T]$.

It is known~\cite{steane1996multiple,calderbank1996good}  that 
a CSS code with check matrices $H^X$ and $H^Z$ has distance
$d\,{=}\,\min{(d^X,d^Z)}$, where 
$d^X$ and $d^Z$ are the code distances for $X$-type and $Z$-type errors
defined as
\[
d^X =  \min\bigl\{ |v|{:} \,\, v\in \ker{H^Z}{\setminus} \rs{H^X} \bigr\}
\quad \mbox{and} \quad
d^Z =  \min\bigl\{ |v|{:} \,\, v\in \ker{H^X}{\setminus} \rs{H^Z} \bigr\}.
\]
We claim that $d^Z\,{\le}\, d^X$.
Indeed, let $X(f)\,{=}\,\prod_{j=1}^n X_j^{f_j}$ be a minimum weight logical $X$-type Pauli operator
such that $|f|\,{=}\,d^X$. 
Then $H^Z f\,{=}\,0$ and $f\,{\notin}\, \rs{H^X}$.
Thus there exists a logical $Z$-type operator $Z(g)\,{=}\,\prod_{j=1}^n Z_j^{g_j}$ anti-commuting with $X(f)$. In other words, $H^X g\,{=}\,0$ and $f^T g \,{=}\, 1$. Here, $f$ and $g$ are length-$n$ binary vectors.
Write $f\,{=}\,(\alpha,\beta)$ and $g\,{=}\,(\gamma, \delta)$, where $\alpha,\beta,\gamma,\delta$ are length-$(n/2)$ vectors.
Conditions $H^Z f =0$ and $H^X g=0$ are equivalent to 
\be
\label{alpha_beta_gamma_delta}
B^T \alpha = A^T \beta \quad \mbox{and} \quad A \gamma = B \delta.
\ee
Here and below all arithmetics is modulo two.
Define length-$n$ vectors 
\be
e = (C\delta, C\gamma) \quad \mbox{and} \quad h=(C \beta, C\alpha).
\ee
From Eqs.~(\ref{ABC2},\ref{alpha_beta_gamma_delta}) one gets
\[
H^X h = [A|B]\left[\ba{c} C \beta \\ C \alpha\\ \ea\right]
=AC\beta+BC\alpha = C(A^T \beta + B^T \alpha) = 0.
\]
Likewise, 
\[
H^Z e = [B^T|A^T]\left[\ba{c} C \delta \\ C \gamma\\ \ea\right]
=B^T C\delta + A^T C\gamma = C(B \delta + A \gamma) = 0.
\]
Furthermore, 
\[
h^T e = \beta^T C C \delta + \alpha^T C C\gamma = \beta^T \delta + \alpha^T \gamma =f^T g=1.
\]
Thus $X(e)$ and $Z(h)$ are non-identity logical operators. 
It follows that $d^Z\le |h|$. We get
\[
d^Z\le |h|=|C\beta|+|C\alpha|=|\beta|+|\alpha|=|f|=d^X.
\]
Thus $d^Z\le d^X$. Similar argument shows that $d^X\le d^Z$, that is, $d^X=d^Z$.
\end{proof}
\noindent
We note that
the equality $d^X=d^Z$ can also be established using the machinery of
Ref.~\cite{panteleev2021quantum} 
by viewing $\qc$ as a Lifted Product code.

In the following, we partition the set of data qubits as $[n]\,{=}\,LR$, where $L$ and $R$ are the left and right blocks of $n/2=\ell m$ data qubits. Then, data qubits $L$ and $R$ and checks $X$ and $Z$ may each be labeled by integers $\mathbb{Z}_{\ell m}=\{0,1,\dots,\ell m-1\}$ which are indices into the matrices $A,B$. Alternatively, qubits and checks can be labeled by monomials from $\mathcal{M} =\{1,y,\dots,y^{m-1},x,xy,\dots,xy^{m-1},\dots,x^{\ell-1}y^{m-1}\}$ in this order, so that $i\in\mathbb{Z}_{\ell m}$ labels the same qubit or check as $x^{a_i}y^{i-ma_i}$ for $a_i=\text{floor}(i/m)$. Using the monomial labeling, $L$ data qubit $\alpha\in \mathcal{M}$ is part of $X$ checks $A^T_i\alpha$ and $Z$ checks $B_i\alpha$ for $i=1,2,3$. Similarly, $R$ data qubit $\beta\in \mathcal{M}$ is part of $X$ checks $B^T_i\beta$ and $Z$ checks $A_i\beta$. A unified notation assigns each qubit or check a label $q(T,\alpha)$ where $T\in\{L,R,X,Z\}$ denotes its type and $\alpha\in \mathcal{M}$ its monomial label\footnote{The monomial notations should not be confused with the matrix notations used earlier in this section. For example, multiplication of monomials such as $B_i\alpha$ is different from multiplying a vector $\alpha$ by a matrix $B_i$.}.

\begin{figure}[ht]
     \centering
    \includegraphics[width=0.9\textwidth]{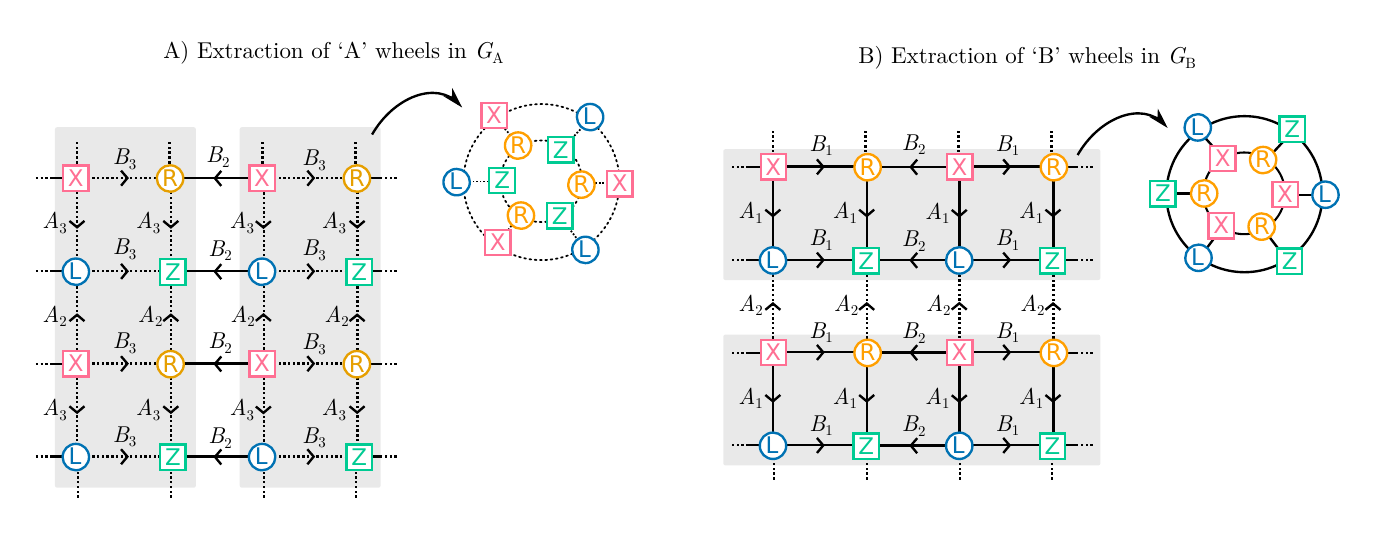}
        \caption{ A). B) Two different grids over a torus defined using different subsets of $A_1,A_2,A_3,B_1,B_2,B_3$.  Edge labels indicate adjacency matrices that generate the respective edges. By extracting either horizontal or vertical strips from these grids, we obtain planar `wheel graphs' whose union contains all edges in the Tanner graph. The `A' wheels (dashed lines) cover $A_2,A_3,B_3$ and the `B' wheels (solid lines) cover $B_1,B_2,A_1$. To avoid clutter, each grid shows only a subset of edges present in the Tanner graph.}
        \label{fig:wheel_extraction}
\end{figure}

\begin{figure}[ht]
     \centering
    \includegraphics[width=0.6\textwidth]{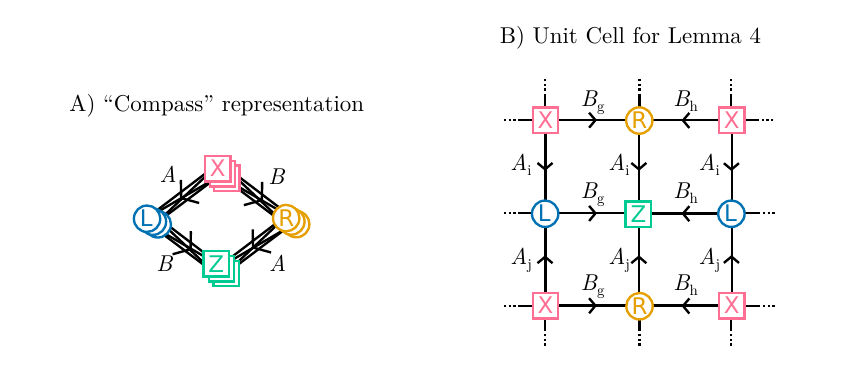}
        \caption{A) ``Compass'' diagram that shows the direction in which matrices $A,B$ are applied to travel between different nodes.  B) The unit cell of the construction of a toric layout in the proof of \lem{toric_layout}.}
        \label{fig:navigation}
\end{figure}

\begin{lemma}
\label{lem:thickness}
The Tanner graph $G$ of the code $\qc$ has thickness $\theta\,{\le}\, 2$. 
A decomposition of $G$ into two planar layers can be computed in time $O(n)$. 
Each planar layer of $G$ is a degree-3 graph.
\end{lemma}
\begin{proof}
Let $G\,{=}\,(V,E)$ be the Tanner graph. 
Partition $G$ into subgraphs $G_A\,{=}\,(V,E_A)$ and $G_B\,{=}\,(V,E_B)$ that describe CSS codes with check matrices
\be
\mbox{Tanner graph $G_A$:} \quad 
 H^X_A=[A_2+A_3|B_3]  \quad \mbox{and} \quad  H^Z_A = [B_3^T|A_2^T+A_3^T]
\ee
\be
\mbox{Tanner graph $G_B$:} \quad 
H^X_B = [A_1 |B_1+B_2] \quad \mbox{and} \quad H^Z_B = [B_1^T+B_2^T|A_1^T].
\ee
Since $A=A_1+A_2+A_3$ and $B=B_1+B_2+B_3$,
every edge of $G$ appears either in $G_A$ or $G_B$, where the two subgraphs are named by whether they contain more $A_i$ edges or more $B_i$ edges. 
Then  $G_A$ and $G_B$ are regular degree-3 graphs (since $A_i$ and $B_j$ are permutation matrices).

Consider the graph $G_A$. 
Each $X$-check vertex  is connected to
a pair of data vertices $i_1,i_2\in L$ via the matrices $A_2, A_3$
and a data vertex $i_3\in R$ via the matrix $B_3$.
Each $Z$-check vertex is connected to a pair of data vertices $i_1,i_2\in R$ via the matrices $A_2^T, A_3^T$ and a data vertex $i_3\in L$
via the matrix $B_3^T$.

We claim that each connected component of $G_A$ can be represented by a ``wheel graph" illustrated in \fig{wheel_extraction}.
A wheel graph consists of two disjoint cycles of the same length $p$
 interconnected by $p$ radial edges.
The outer cycle alternates between $X$-check 
and $L$-data vertices.

Edges of the outer cycle alternate between those generated by $A_3$ 
(as one moves from a check  to a data vertex)
and $A_2^T$ (as one moves from a data to a check vertex).
The length of the outer cycle is equal to the order of the matrix $A_3 A_2^T$,
that is, the smallest integer $p$ such that $(A_3 A_2^T)^p=I_{\ell m}$.
For example, consider the code $[[144,12,12]]$ from \tab{codes}.
Then $A=x^3+y+y^2$, $A_2=y$, and $A_3=y^2$.
Thus $A_3 A_2^T = y^2 y^{-1} = y$ which has order $m =6$.
The inner cycle of a wheel graph alternates between $Z$-check and $R$-data vertices.

Edges of the inner cycle  alternate between those generated by $A_3^T$
(as one moves from a check to a data vertex) and $A_2$ (as one moves from a data to a check vertex). 
The length of the inner cycle is equal to the order of the matrix $A_3^T A_2$ which is just the transpose of $A_3A_3^T$ considered earlier. Thus both inner and outer cycles have the same length $m$. The two cycles are interconnected by $m$ radial edges as shown in \fig{wheel_extraction} A). 
Radial edges are generated by the matrix $B_3$,
as one moves  towards the center of the wheel. 
The wheel graph contains $4$-cycles
generated by tuples of edges $(B_3,A_2,B_3^T,A_2^T)$ and
$(B_3^T,A_3,B_3,A_3^T)$.
Commutativity between $A_i$ and $B_j$
ensures that traversing any of these $4$-cycles implements the identity matrix, that is,
the graph is well defined. 
Clearly, the wheel graph is planar. Since $G_A$ is a disjoint union of wheel graphs, $G_A$ is planar. 
The same argument shows that $G_B$ is planar: see \fig{wheel_extraction} B).
\end{proof}

We empirically observed that  BB codes reported in \tab{codes} have no weight-$4$ stabilizers.
The presence of
such stabilizers
is known to have a negative impact on the performance
of belief propagation decoders~\cite{panteleev2021degenerate}, which we use here.

The definition of code $\qc$ does not guarantee that its Tanner graph is connected. Some choices of $A$ and $B$ lead to a code that is actually several separable code blocks. This manifests as a Tanner graph with several connected components. For instance, although all codes in \tab{codes} are connected, taking any of them with even $\ell$ and replacing every instance of $x$ with $x^2$ creates a code with two connected components.

\begin{lemma} \label{lem:connected}
The Tanner graph of the code $\qc$ is connected if and only if $S=\{A_iA_j^T:i,j\in\{1,2,3\}\}\cup\{B_iB_j^T:i,j\in\{1,2,3\}\}$ generates the group $\mathcal{M}$. The number of connected components in the Tanner graph is $\ell m/|\langle S\rangle|$, and all components are graph isomorphic to one another.
\end{lemma}
\begin{proof}
\fig{navigation} is helpful for following the arguments in this proof. We start by proving the reverse implication of the first statement. Note that there is a length 2 path in the Tanner graph from $L$ qubit $\alpha\in \mathcal{M}$ to $L$ qubit $A_iA_j^T\alpha$ and another length 2 path to $L$ qubit $B_iB_j^T\alpha$. These travel through $X$ and $Z$ checks, respectively. Thus, because the $A_iA_j^T$ and $B_iB_j^T$ generate $\mathcal{M}$, there is some path from $\alpha$ to any other $L$ qubit $\beta$. 
A similar argument shows the existence of a path connecting any pair of R qubits. Since each $X$ check and each $Z$ check are connected to at least one $L$ qubit and at least one $R$ qubit, this implies that the entire Tanner graph is connected.
The forward implication of the first statement follows after noticing that, for all $T\in\{L,R,X,Z\}$, the path from a type $T$ node to any other $T$ node is necessarily described as a product of elements from $S$. Connectivity of the Tanner graph implies the existence of all such paths, and so $S$ must generate $\mathcal{M}$. 

If $S$ does not generate $\mathcal{M}$, it necessarily generates a subgroup $\langle S\rangle$ and nodes in connected components of the Tanner graph are labeled by elements of the cosets of this subgroup. This implies the theorem's second statement.
\end{proof}

For the next part, we establish some terminology. A spanning sub-graph of a graph $G$ is a sub-graph containing all the vertices of $G$. Also, the undirected Cayley graph of a finite Abelian group $\mathcal{G}$ (with identity element $0$) generated by set $S\subset\mathcal{G}$ is the graph with vertex set $\mathcal{G}$ and undirected edges $(g,g+s)$ for all $g\in\mathcal{G}$ and all $s\in S,s\neq0$. We say the Cayley graph of $\mathbb{Z}_{a}\times\mathbb{Z}_{b}$ when we mean the Cayley graph of $\mathbb{Z}_{a}\times\mathbb{Z}_{b}$ generated by $\{(1,0),(0,1)\}$. The order $\ord{g}$ of an element $g$ in a multiplicative group is the smallest positive integer such that $g^{\ord{g}}=1$. 

\begin{dfn}
Code $\qc$ is said to have a toric layout if its Tanner graph has a spanning sub-graph isomorphic to the Cayley graph of $\mathbb{Z}_{2\mu}\times\mathbb{Z}_{2\lambda}$ for some integers $\mu$ and $\lambda$.
\end{dfn}

Note that only codes with connected Tanner graphs can have a toric layout according to this definition. An example toric layout is depicted in \fig{2Dlayout}~B).

\begin{lemma}\label{lem:toric_layout}
A code $\qc$ has a toric layout if there exist $i,j,g,h\in\{1,2,3\}$ such that
\begin{enumerate}[label=(\roman*)]
\item $\langle A_iA_j^T,B_gB_h^T\rangle=\mathcal{M}$ and
\item $\ord{A_iA_j^T}\ord{B_gB_h^T}=\ell m$.
\end{enumerate}
\end{lemma}
\begin{proof}
We let $\mu=\ord{A_iA_j^T}$ and $\lambda=\ord{B_gB_h^T}$. We associate qubits and checks in the Tanner graph of $\qc$ with elements of $\mathcal{G}=\mathbb{Z}_{2\mu}\times\mathbb{Z}_{2\lambda}$. For $L$ qubit with label $\alpha\in \mathcal{M}$, because of (i), there is $(a,b)\in\mathbb{Z}_{\mu}\times\mathbb{Z}_{\lambda}$ such that $\alpha=(A_iA_j^T)^a(B_gB_h^T)^b$. Because of (ii) and the pigeonhole principle, this choice of $(a,b)$ is unique. We associate $L$ qubit $\alpha$ with $(2a,2b)\in\mathcal{G}$. Similarly, an $R$ qubit with label $\alpha A^T_jB_g$ is associated with $(2a+1,2b+1)\in\mathcal{G}$, $X$-check $\alpha A^T_j$ with $(2a+1,2b)$, and $Z$-check $\alpha B_g$ with $(2a,2b+1)$. Edges in the Tanner graph $A_i,A^T_j,B_g,$ and $B_h^T$ can now be drawn as in \fig{navigation} (B) and correspond to edges in the Cayley graph of $\mathcal{G}$. For instance, to get from $(2a+1,2b+1)$, an $R$ qubit, to $(2a+2,2b+1)$, a $Z$ check, we apply $A_i$, taking $R$ qubit labeled $\alpha A^T_jB_g$ to the $Z$ check labeled $(\alpha A^T_jB_g)A_i=\alpha(A_iA_j^T)B_g$.
\end{proof}

All codes in \tab{codes} have a toric layout with $\mu=m$ and $\lambda=\ell$. Most of these codes satisfy \lem{toric_layout} with $i=g=2$ and $j=h=3$. The exception is the $\llbracket90,8,10\rrbracket$ code, for which we can take $i=2,g=1$ and $j=h=3$.

However, we also note two interesting cases. First, there are codes with connected Tanner graphs that do not satisfy the conditions for a toric layout given in \lem{toric_layout}. 
One example of such a code is
$\qc$
with $\ell,m=28,14$,
$A=x^{26}+y^6 + y^8$, and
$B=y^7+x^9+x^{20}$
which has parameters $[[784,24,\le 24]]$.
Second, for a code satisfying the conditions of \lem{toric_layout}, it need not be the case that the set $\{\ord{A_iA_j^T},\ord{B_gB_h^T}\}$ and the set $\{\ell,m\}$ are equal. For example, the $[[432,4,\le22]]$ code with $\ell,m=18,12$ and $A=x+y^{11}+y^3$, $B=y^2+x^{15}+x$ only satisfies \lem{toric_layout} with $\mu,\lambda=36,6$ (take $i=g=1$ and $j=h=2$ for instance). 

\newcommand{\maindocument}{}
\ifdefined\maindocument
\else
    \documentclass{article}
    \input{packages.tex}

\title{Supplemental Material for\\ ``High-threshold and low-overhead fault-tolerant quantum memory''}

 \author[1]{Sergey~Bravyi}
 \author[1]{Andrew~W.~Cross}
 \author[1]{Jay~M.~Gambetta}
  \author[1]{Dmitri~Maslov}
 \author[2]{Patrick~Rall}
 \author[1]{Theodore~J.~Yoder} 
 \affil[1]{\large{IBM Quantum, IBM T.J. Watson Research Center, Yorktown Heights, NY 10598 (USA)}}
 \affil[2]{IBM Quantum, MIT-IBM Watson AI Lab, Cambridge, MA 02142 (USA)}

    \begin{document}
    \maketitle
    
    \tableofcontents
    
   \newcommand{\figlayout}{1 }
   \newcommand{\figwheelextraction}{4 }
   \newcommand{\eqnAB}{1}
   \newcommand{\lemkd}{1 }
   \newcommand{\figsyndromecircuit}{2 }
   \newcommand{\figsmallldpc}{3 A) }
   \newcommand{\figldpcvssurface}{3 B) }
   \newcommand{\tabcodesintro}{1 }
   \newcommand{\tabcodes}{2 }
   \newcommand{\fignavigation}{5 }
   \newcommand{\lemconnected}{3 }

\fi

\newcommand{\ifmain}[2]{\ifdefined\maindocument#1\else#2\fi}


\section{Syndrome measurement circuit}
\label{sec:syndrome_circuit}

The next step is to furnish the code $\qc$ with a syndrome measurement (SM) circuit that repeatedly measures the syndrome of each check operator.  Here we describe a SM circuit that requires $2n$ physical qubits in total: $n$ data qubits and $n$ ancillary check qubits used to record the measured syndromes.  The circuit only applies $\cnotgate$s to pairs of qubits that are connected in the Tanner graph. 

The SM circuit is defined as a periodically repeated sequence of {\em syndrome cycles} (SC).  A single SC is responsible for measuring syndromes of all $n$ check operators of the code.  Let $N_c$ be the number of syndrome cycles. We envision that $N_c\,{>}\,1$.  The circuit begins and ends with a special initialization and measurement cycle responsible for initializing logical qubits in a suitable initial state and measuring logical qubits in a suitable basis.  Here we focus on the optimization of the SC circuit.  Logical initialization and measurements are discussed in \sec{logicals}.

The SC circuit is divided into $N_r$ {\em rounds} such that each round is a depth-1 circuit composed of $\cnotgate$s and single-qubit operations.  The latter include initializing a qubit in the $X$ or $Z$ basis and measuring a qubit in the $X$ or $Z$ basis.  $\cnotgate$s can be applied only to pairs of qubits which are nearest neighbors in the Tanner graph.  Some qubits remain idle during some rounds, although we try to minimize such occurrences by squeezing more useful computations in as little time as possible.  Our notations are summarized in \tab{tab3}.

\begin{table}[t]
\begin{center}
\begin{tabular}{c|c}
\hline
Notation & Operation  \\
\hline
\hline
$\cnot{c}{t}$  & $\cnotgate$ with control qubit $c$ and target qubit $t$ \\
\hline
$\initX{q}$  & Initialize qubit $q$ in the state $|+\ra=(|0\ra+|1\ra)/\sqrt{2}$\\
\hline
$\initZ{q}$ &  Initialize qubit $q$ in the state $|0\ra$\\
\hline
$\measX{q}$ & Measure qubit $q$ in the $X$-basis $|+\ra,|-\ra$\\
\hline
$\measZ{q}$ & Measure qubit $q$ in the $Z$-basis $|0\ra,|1\ra$\\
\hline
$\idle{q}$ & Identity gate on qubit $q$  \\ 
\hline
\end{tabular}
\caption{Elementary operations used for syndrome measurements.}
\label{table:tab3}
\end{center}
\end{table}

\begin{table}[ht]
\begin{center}
\begin{tabular}{c|c||c|c}
\hline
Round  & Circuit & Round    & Circuit \\
\hline
\hline
1 & 
{\centering
\begin{minipage}{7cm}
	\begin{algorithmic}
	\State{}
	\For{$i=1$ to $n/2$}
	\State{$\initX{q(X,i)}$}
	\State{$\cnot{q(R,A_1^T(i))}{q(Z,i)}$}
	\State{$\idle{q(L,i)}$}
	\EndFor
	\State{}
	\end{algorithmic}
\end{minipage}}
&  5 &
{\centering
\begin{minipage}{7cm}
	\begin{algorithmic}
	\State{}
	\For{$i=1$ to $n/2$}
	\State{$\cnot{q(X,i)}{q(R,B_3(i))}$}
	\State{$\cnot{q(L,B_3^T(i))}{q(Z,i)}$}
	\EndFor
	\State{}
	\end{algorithmic}
\end{minipage}}
 \\
\hline
2 & 
{\centering
\begin{minipage}{7cm}
	\begin{algorithmic}
	\State{}
	\For{$i=1$ to $n/2$}
	\State{$\cnot{q(X,i)}{q(L,A_2(i))}$}
	\State{$\cnot{q(R,A_3^T(i))}{q(Z,i)}$}
	\EndFor
	\State{}
	\end{algorithmic}
\end{minipage}}
& 6 & 
{\centering
\begin{minipage}{7cm}
	\begin{algorithmic}
	\State{}
	\For{$i=1$ to $n/2$}
	\State{$\cnot{q(X,i)}{q(L,A_1(i))}$}
	\State{$\cnot{q(R,A_2^T(i))}{q(Z,i)}$}
	\EndFor
	\State{}
	\end{algorithmic}
\end{minipage}}
\\
\hline
3 & 
{\centering
\begin{minipage}{7cm}
	\begin{algorithmic}
	\State{}
	\For{$i=1$ to $n/2$}
	\State{$\cnot{q(X,i)}{q(R,B_2(i))}$}
	\State{$\cnot{q(L,B_1^T(i))}{q(Z,i)}$}
	\EndFor
	\State{}
	\end{algorithmic}
\end{minipage}}
& 7 &
{\centering
\begin{minipage}{7cm}
	\begin{algorithmic}
	\State{}
	\For{$i=1$ to $n/2$}
	\State{$\cnot{q(X,i)}{q(L,A_3(i))}$}
	\State{$\measZ{q(Z,i)}$}
	\State{$\idle{q(R,i)}$}
	\EndFor
	\State{}
	\end{algorithmic}
\end{minipage}}
 \\
\hline
4 & 
{\centering
\begin{minipage}{7cm}
	\begin{algorithmic}
	\State{}
	\For{$i=1$ to $n/2$}
	\State{$\cnot{q(X,i)}{q(R,B_1(i))}$}
	\State{$\cnot{q(L,B_2^T(i))}{q(Z,i)}$}
	\EndFor
	\State{}
	\end{algorithmic}
\end{minipage}}
& 8 & 
{\centering
\begin{minipage}{7cm}
	\begin{algorithmic}
	\State{}
	\For{$i=1$ to $n/2$}
	\State{$\measX{q(X,i)}$}
	\State{$\initZ{q(Z,i)}$}
	\State{$\idle{q(L,i)}$}
	\State{$\idle{q(R,i)}$}
	\EndFor
	\State{}
	\end{algorithmic}
\end{minipage}}
\\
\hline
\end{tabular}
\end{center}
\caption{Depth-8 syndrome measurement cycle circuit.}
\label{table:syndromecircuit}
\end{table}

Below we describe a SC circuit with effectively $N_r\,{=}\,8$ rounds\footnote{The operator $\initZ{q(Z,i)}$ must be executed before the first application of this SC circuit, raising the depth of the first stage to 9.  However, each following syndrome cycle takes Z-check initialization from the previous round. Last syndrome cycle needs not apply the Z-check state initialization.  Total SM circuit depth with $N_c$ syndrome cycles is thus $8N_c\,{+}\,1$.}.  Ignoring single-qubit initialization and measurement operations, the SC circuit is a depth-7 $\cnotgate$ circuit.  By designing the circuit for an explicit family of LDPC codes we are able to leverage the symmetries and reduce computational depth to $7$ from what otherwise would be $14\,{=}\,2{\cdot}6\,{+}\,2$, as shown by previous authors \cite[Theorem 1]{tremblay2022constant}.  Our notations are as follows.  We divide $n$ data qubits into the left and the right registers $q(L)$ and $q(R)$ of size $n/2$ each.  Each check operator acts on three data qubits from $q(L)$ and three data qubits from $q(R)$.  The SM circuit uses $2n$ physical qubits in total: $n$ data qubits and $n$ ancillary check qubits that record  the syndrome of each check operator.  Let $q(X)$ and $q(Z)$ be the ancillary registers of size $n/2$ that span $X$-check and $Z$-check qubits respectively.
Thus the physical qubits are partitioned into four registers,  $q(X)$, $q(L)$, $q(R)$, and $q(Z)$, of size $n/2$ each. Label qubits in each register by integers $i=1,2,\ldots,n/2$.
We write $q(X,i)$ for the $i$-th qubit of the register $q(X)$ with similar notations for $q(L)$, $q(R)$, and $q(Z)$.
Each permutation matrix $A_p$ and $B_q$ from \ifmain{Eq.~(\ref{AB})}{Eq.~(\eqnAB) in the main text} defines a one-to-one map from the set $\{1,2,\ldots,n/2\}$ onto itself.

We identify a permutation matrix and the corresponding one-to-one map. For example, we write $j\,{=}\,A_1(i)$ if the matrix $A_1$ has a one at row $i$ and column $j$ (this is well defined since $A_1$ is a permutation matrix).
Likewise, we write $j\,{=}\,A_1^T(i)$ if the transposed matrix $A_1^T$ has a one at the row $i$ and column $j$.
In this notation, the $i$-th $X$-check operator acts on data qubits $q(L,A_p(i))$ and $q(R,B_p(i))$ with $p=1,2,3$.  The $i$-th $Z$-check operator acts on data qubits $q(L,B_p^T(i))$ and $q(R,A_p^T(i))$ with $p=1,2,3$.

\ifdefined\maindocument
\begin{figure}[t]
  \centering
    \includegraphics[height=8cm]{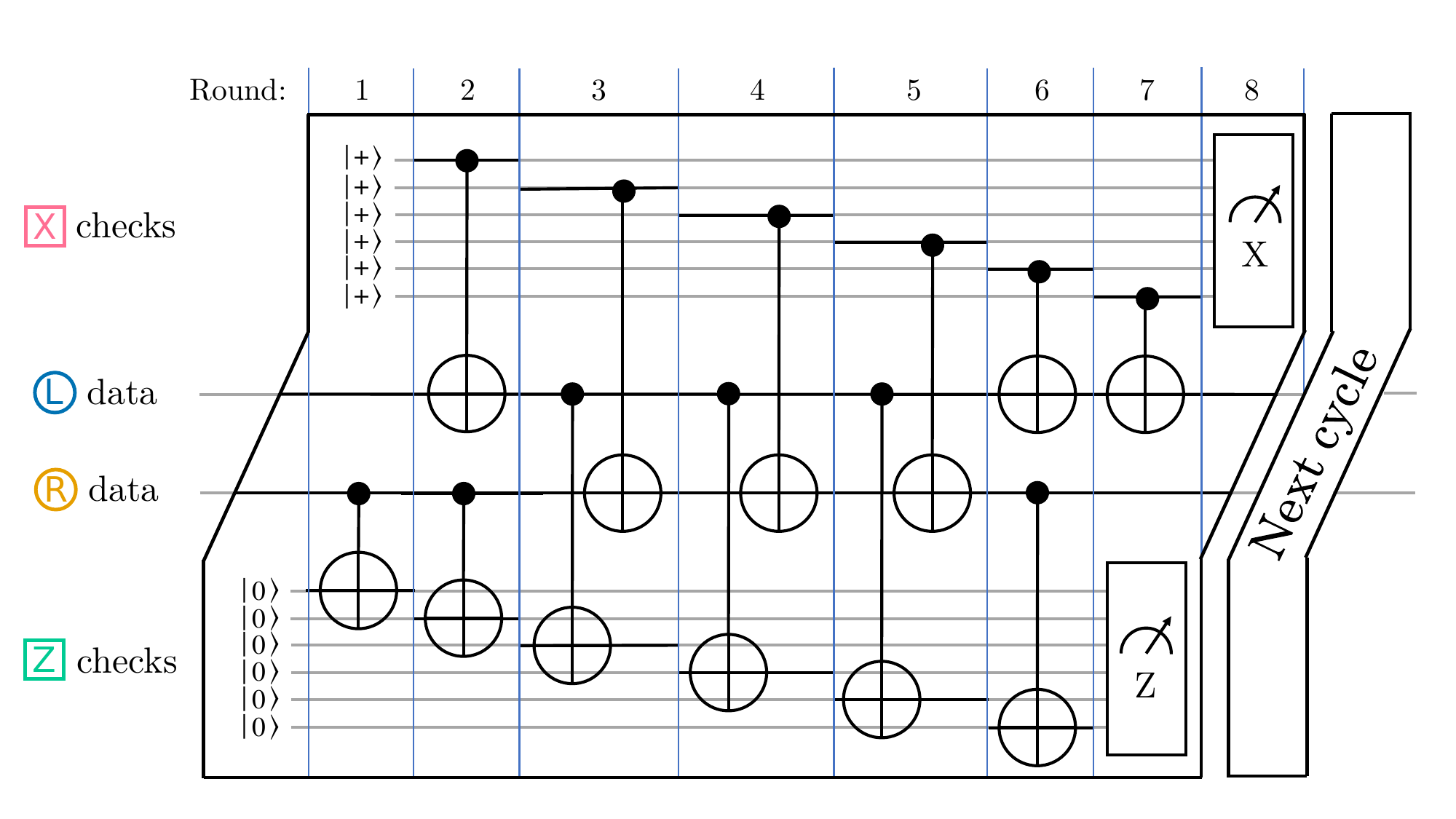}
     \caption{Depth-8 syndrome measurement cycle circuit.}
  \label{fig:syndromecircuit}
\end{figure}
\else
\fi

Our depth-8 SC circuit is described in \tab{syndromecircuit} and illustrated in \ifmain{\fig{syndromecircuit}}{Figure~\figsyndromecircuit in the main text}.  Note that within each round all operations act over non-overlapping sets of qubits.  In particular, each round applies at most one layer of $\cnotgate$ gates between $q(X)$ and $q(L)$ registers (Rounds 2, 6, and 7), at most one layer of $\cnotgate$s between $q(X)$ and $q(R)$ registers (Rounds 3, 4, and 5), at most one layer of $\cnotgate$s between $q(Z)$ and $q(L)$ registers (Rounds 3, 4, and 5), and at most one layer of $\cnotgate$s between $q(Z)$ and $q(R)$ registers (Rounds 1, 2, and 6). Qubits from $q(Z)$ are always targets for $\cnotgate$s.  Accordingly, $X$-type errors propagate from data qubits to check qubits in $q(Z)$.  The latter are measured in the $Z$-basis in Round 7 revealing the syndrome of $X$-type errors. Qubits from $q(X)$ are always controls for $\cnotgate$s. Accordingly, $Z$-type errors propagate from data qubits to check qubits in $q(X)$.  The latter are measured in the $X$-basis in Round~8 revealing the syndrome of $Z$-type errors.  We envision that the syndrome cycles are repeated periodically. This justifies applying $\cnotgate$s to $q(Z)$ at Round~1 even though $q(Z)$ is initialized only at Round~8.  Indeed, Round~8 of the previous syndrome cycle goes immediately before Round~1 of the current cycle. Thus $q(Z)$ has been already initialized at the beginning of Round~1.  We were not able to find a depth-8 (or smaller depth) syndrome cycle in which $X$-check and $Z$-check qubits are initialized and measured synchronously. 

Let us now prove that the above SC circuit has the desired functionality.  Since the circuit involves only Clifford operations, its action can be compactly described using stabilizer tableau~\cite{aaronson2004improved}.  We track how the tableau changes as each layer of $\cnotgate$s in the circuit is applied.  Since the $\cnotgate$ gates do not mix Pauli $X$ and $Z$ operators, one may consider tableau describing the action of the circuit on $X$-type and $Z$-type Pauli operators separately.  

Let us begin with $X$-type Pauli.  The corresponding tableau $T$ is a binary matrix of size $n\,{\times}\, 2n$ such that each row of $T$ defines an $X$-type stabilizer of the underlying quantum state.  We partition columns of $T$ into four blocks that represent qubit registers $q(X)$, $q(L)$, $q(R)$, and $q(Z)$.  We partition rows of $T$ into two blocks such that initially the top $n/2$ rows represent weight-$1$ check operators on qubits of the register $q(X)$ initialized in the state $|+\ra$ while the bottom $n/2$ rows represent weight-$6$ check operators on data qubits associated with the chosen code $\qc$. 
Thus, at the beginning of Round~1, when all check qubits in the register $q(X)$ have been initialized in the state $|+\ra$, while data qubits are in some logical state of the code $\qc$, the binary matrix is 
\[
\begin{pmatrix}
I & 0 & 0 & 0 \\
0 & A & B & 0
\end{pmatrix}.
\]
Here $I\,{\equiv}\, I_{n/2}$ is the identity matrix. 
The SC circuit (ignoring qubit initialization and measurements) enacts the transformation
\[
\begin{pmatrix}
I & 0 & 0 & 0 \\
0 & A & B & 0
\end{pmatrix}
\xrightarrow{\text{SC circuit}}
\begin{pmatrix}
I & A & B & 0 \\
0 & A & B & 0 
\end{pmatrix}.
\]
Indeed, the circuit must map a single-qubit $X$ stabilizer $X_j$ on a check qubit $j\,{\in}\, q(X)$ to a product of $X_j$ and the $j$-th $X$-type check operator on the data qubits determined by the $j$-th row of $H^X=[A|B]$.
The eigenvalue measurement of $X_j$ at the final round then reveals the syndrome of the $j$-th check operator.
The bottom $n/2$ rows must be unchanged since the check operators of the code must be the same before and after the syndrome measurement.

Let us verify that the circuit defined in \tab{syndromecircuit} enacts the desired transformation. To accomplish this, we rewrite the SC circuit by removing notations irrelevant to showing the correctness of $X$-checks.  Specifically, we write each $\cnotgate$ in \tab{syndromecircuit} as $\cnotgate_M(a,b)$, where $a,b \,{\in}\, \{1, 2, 3, 4\} = \{q(X), q(L), q(R), q(Z)\}$, and $M \,{\in}\, \{A_1,A_2,A_3,B_1,B_2, B_3\}$. Note that the $\cnotgate$ instructions where the matrix $M^T$ is used instead of $M$ can be written using matrix $M$ by performing the variable renaming $i\,{\gets}\, M(i)$ in the corresponding for loop in \tab{syndromecircuit}.

Using the above compact notation, the unitary part of the SC circuit becomes:
\be
\ba{cc}
\text{Round~1:} & \hfill \cnotgate_{A_1}(3,4) \\
\text{Round~2:} & \cnotgate_{A_2}(1,2), \cnotgate_{A_3}(3,4) \\
\text{Round~3:} & \cnotgate_{B_2}(1,3), \cnotgate_{B_1}(2,4) \\
\text{Round~4:} & \cnotgate_{B_1}(1,3), \cnotgate_{B_2}(2,4) \\
\text{Round~5:} & \cnotgate_{B_3}(1,3), \cnotgate_{B_3}(2,4) \\
\text{Round~6:} & \cnotgate_{A_1}(1,2), \cnotgate_{A_2}(3,4) \\
\text{Round~7:} & \cnotgate_{A_3}(1,2) \hfill \\
\ea
\label{SCunitary_part}
\ee

In the following we apply all seven unitary rounds to verify the correctness of the performed transformation:
\[
\text{Round~1:}
\begin{pmatrix}
I & 0 & 0 & 0 \\
0 & A & B & 0
\end{pmatrix}
\xrightarrow{\cnotgate_{A_1}(3,4)}
\begin{pmatrix}
I & 0 & 0 & 0 \\
0 & A & B & A_1 B
\end{pmatrix}
\]

\[
\text{Round~2:}
\xrightarrow{\cnotgate_{A_2}(1,2)}
\begin{pmatrix}
I & A_2 & 0 & 0 \\
0 & A & B & A_1 B
\end{pmatrix}
\xrightarrow{\cnotgate_{A_3}(3,4)}
\begin{pmatrix}
I & A_2 & 0 & 0 \\
0 & A & B & (A_1{+}A_3)B \\
\end{pmatrix}
\]

\[
\text{Round~3:}
\xrightarrow{\cnotgate_{B_2}(1,3)}
\begin{pmatrix}
I & A_2 & B_2 & 0 \\
0 & A & B & (A_1{+}A_3)B
\end{pmatrix}
\xrightarrow{\cnotgate_{B_1}(2,4)}
\begin{pmatrix}
I & A_2 & B_2 & A_2B_1 \\
0 & A & B & (A_1{+}A_3)B+AB_1 \\
\end{pmatrix}
\]

\[
\text{Round~4:}
\xrightarrow{\cnotgate_{B_1}(1,3)}
\begin{pmatrix}
I & A_2 & B_1{+}B_2 & A_2 B_1 \\
0 & A & B & (A_1{+}A_3)B + AB_1 \\
\end{pmatrix}
\]
\[
\xrightarrow{\cnotgate_{B_2}(2,4)}
\begin{pmatrix}
I & A_2 & B_1{+}B_2 & A_2 (B_1{+}B_2) \\
0 & A & B & (A_1{+}A_3)B+A(B_1{+}B_2) \\
\end{pmatrix}
=
\begin{pmatrix}
I & A_2 & B_1{+}B_2 & A_2 (B_1{+}B_2) \\
0 & A & B & A_2B + AB_3 \\
\end{pmatrix}
\]
Here, we use the identity $(A_1{+}A_3)B+A(B_1{+}B_2)=A_2B + AB_3$, which holds since the sum of first summands and second summands on both sides of the equation gives $AB$, and $AB\,{+}\,AB\,{=}\,0$.

\[
\text{Round~5:}
\xrightarrow{\cnotgate_{B_3}(1,3)}
\begin{pmatrix}
I & A_2 & B & A_2 (B_1{+}B_2) \\
0 & A & B & A_2B + AB_3  \\
\end{pmatrix}
\xrightarrow{\cnotgate_{B_3}(2,4)}
\begin{pmatrix}
I & A_2 & B & A_2 B \\
0 & A & B & A_2B   \\
\end{pmatrix}
\]

\[
\text{Round~6:}
\xrightarrow{\cnotgate_{A_1}(1,2)}
\begin{pmatrix}
I & A_1{+}A_2 & B & A_2 B \\
0 & A & B & A_2B   \\
\end{pmatrix}
\xrightarrow{\cnotgate_{A_2}(3,4)}
\begin{pmatrix}
I & A_1{+}A_2 & B & 0 \\
0 & A & B & 0  \\
\end{pmatrix}
\]

\[
\text{Round~7:}
\xrightarrow{\cnotgate_{A_3}(1,2)}
\begin{pmatrix}
I & A & B & 0 \\
0 & A & B & 0  \\
\end{pmatrix}.
\]

\noindent This is the desired transformation.

So far, we have not considered the action of the SC circuit on the logical qubits of the code.  Let us show that this action is trivial.  Indeed, consider some $X$-type logical operator $X(v)$, where $v\,{\in}\, \FF_2^n$. Write $v\,{=}\,(u,w)$ where $u$ and $w$ are restrictions of $v$ onto the registers 2 and 3 respectively. Commutativity between $X(v)$ and any $Z$-type check operator implies 
\[
uB+wA=0.
\]
Here we consider $u$ and $w$ as row vectors.  Extending $v$ by zeroes on registers 1 and 4 gives the row vector $(0\; u \; w\; 0)$, where $0$ stands for the all-zero row vector of length $n/2$.  Let us follow the same chain of transformations as above starting from the initial vector $(0\; u \; w\; 0)$. All $\cnotgate$s controlled by the register 1, such as $\cnotgate_{A_2}(1,2)$ or $\cnotgate_{B_2}(1,3)$ in Eq.~(\ref{SCunitary_part}), have trivial action on the vector $(0\; u \; w\; 0)$ since all qubits of the control register are zeroes.  Such $\cnotgate$s can be omitted.  The remaining $\cnotgate$s in Eq.~(\ref{SCunitary_part}) such as $\cnotgate_{A_1}(3,4)$ or $\cnotgate_{B_1}(2,4)$ map the initial vector $(0\; u \; w\; 0)$ to $(0\; u \; w\; t)$ for some vector $t$ since the registers 2 and 3 always serve as the controls and the register 4 always serves as the target.  Rounds~1, 2, and 6 in Eq.~(\ref{SCunitary_part})  are equivalent to XORing vectors $wA_1$, $wA_3$, and $wA_2$ respectively to the register 4.  Rounds~3, 4, and 5 in Eq.~(\ref{SCunitary_part}) are equivalent to XORing vectors $uB_1$, $uB_2$, and $uB_3$ respectively to the register 4.  Thus
\[
t = w(A_1{+}A_2{+}A_3) + u(B_1{+}B_2{+}B_3) = wA + uB = 0.
\]
We have shown that the SC circuit maps the vector  $(0\; u \; w\; 0)$ to itself. Hence the circuit acts trivially on logical $X$-type operators. 

To prove the correctness of $Z$-checks, observe that $Z$-checks can be mapped into $X$-checks by conjugation with Hadamards.  When the unitary circuit in \ifmain{\fig{syndromecircuit}}{Figure~\figsyndromecircuit in the main text} is conjugated with Hadamards, this flips controls and targets of all $\cnotgate$ gates.  Thus, to verify $Z$-checks, it suffices to perform a very similar calculation to the one already shown for $X$-checks.  We omit this calculation here.


The SC circuit shown in \tab{syndromecircuit} is not unique in the following sense: we found 935 depth-7 alternatives to the unitary part of the SC circuit via a computer search. 
These alternatives are obtained from the circuit defined in  Eq.~(\ref{SCunitary_part}) by applying the gate layers $\cnotgate_{A_i}$ and $\cnotgate_{B_j}$ in a different order. 
In the special case of the $[[144,12,12]]$ code,
numerical simulations show that all 936 variants of the syndrome cycle give rise to syndrome measurement circuits with distance $\dcirc\le 10$ explaining our focus on a specific circuit Eq.~(\ref{SCunitary_part}) which we conjecture to have distance $\dcirc\,{=}\,10$.  The short depth of the single cycle, relying on only seven computational stages, helps to keep the spread of errors under control.  Details of calculating upper bounds on the circuit-level distance are provided in \sec{decoder}.

\section{Decoder for the circuit-based noise model}
\label{sec:decoder}

So far we assumed that  the SM circuit is noiseless.
As shown in \sec{syndrome_circuit},
in this case all measured syndromes are zero and the circuit implements the logical identity gate.
Consider now what happens when each operation in the circuit including $\cnotgate$ gates, qubit initializations,
measurements, and idle qubits is subject to noise. To enable efficient decoding and
numerical simulations, we use
the standard circuit-based depolarizing noise model~\cite{fowler2009high}.
It assumes that each operation in the circuit is  ideal or faulty with the
probability $1-p$ or $p$ respectively. Here $p$ is a model parameter
called the error rate. Faults on different operations occur independently.
We define faulty operations as follows.
A faulty $\cnotgate$  is an ideal $\cnotgate$
followed by one of $15$ non-identity Pauli errors on the control and the target qubits
 picked uniformly at random. A faulty initialization
prepares a single-qubit state orthogonal to the correct one.
A faulty measurement is an ideal measurement followed by a classical
bit-flip error applied to the measurement outcome.
A faulty idle  qubit suffers from a Pauli error
$X$ or $Y$ or $Z$ picked uniformly at random.

To perform error correction one needs a decoder — a classical algorithm that takes as input
the measured  error syndrome and outputs a guess of the final Pauli error on the data qubits resulting from all
faults in the SM circuit. The error syndrome may itself be faulty
due to measurement errors. The decoder succeeds if the guessed Pauli error 
coincides with the actual error up to  a product of check operators.
In this case the guessed and the actual error have the same action on any logical state.

Let us show how to adapt Belief Propagation with an Ordered Statistics 
postprocessing step
Decoder (BP-OSD) proposed in~\cite{panteleev2021degenerate,roffe2020decoding}
to the circuit-based noise model.
The decoder consists of two stages.
The first stage takes as input a BB
code $\qc$ equipped with a SM circuit $\calU$  and an error rate $p$. 
It outputs a certain linearized noise model 
that ignores possible cancellations between errors generated by two or more faulty
operations in $\calU$.  This stage is analogous to computing
the decoding graph in error correction algorithms based on the surface code~\cite{fowler2012topological,higgott2023sparse}.
The second (online) stage of the decoder
 takes as input an
error syndrome measured in the experiment
and outputs a guess of the final error on the data qubits. 
This stage
 decodes the linearized noise model
using BP-OSD method~\cite{panteleev2021degenerate,roffe2020decoding}.
\edit{Our linearized noise model is conceptually similar to 
spacetime codes studied by Delfosse and Paetznick~\cite{delfosse2023spacetime}
and detector-based noise model proposed
by McEwen, Bacon, and Gidney~\cite{mcewen2023relaxing}. The online stage of our decoder closely follows Refs.~\cite{higgott2023improved,geher2023tangling}. In particular, Geh\'er, Crawford, and Campbell~\cite{geher2023tangling} applied BP-OSD to study tangled syndrome measurement circuits capable of measuring certain non-local check operators on a hardware with short-range
qubit connectivity.
Higgott et al~\cite{higgott2023improved} showed that the performance of the standard minimum-weight matching decoder can be enhanced by computing  prior error probabilities using BP-decoder as a preprocessing step. 
 }

We begin by describing the offline stage.
Consider a BB code  with parameters $[[n,k,d]]$ and let
$\calU$ be the SM
circuit 
constructed  in \sec{syndrome_circuit} with $N_c$ syndrome cycles.
The circuit $\calU$ contains $6nN_c$ $\cnotgate$s, $nN_c$ initializations and measurements, and $2n N_c$ idle qubit locations.
Let $\calU_1,\calU_2,\ldots,\calU_M$ be the list of all possible faulty realizations of $\calU$
with  exactly one faulty operation. If the faulty operation happens to be $\cnotgate$ or an idle qubit, 
one of the admissible Pauli errors for this operation is specified. 
A simple counting shows that $M=98nN_c$, where $98=15\cdot 6 + 1 + 1 + 3\cdot 2$
accounts for $15$ noisy realizations of each $\cnotgate$, $3$ realizations of memory errors
on idle qubits, noisy initializations and measurements.  
By definition, the list $\calU_1,\calU_2,\ldots,\calU_M$ includes all realizations of $\calU$
that can occur with the probability $O(p)$ in the limit $p\to 0$.
We simulate each circuit $\calU_j$ by propagating the chosen Pauli error towards the final time step
 taking into account qubit initialization and measurement errors (if any).
This simulation can be performed efficiently using the stabilizer formalism.
Let $s_j^U\in \{0,1\}^{nN_c}$ be the full measured syndrome of $\calU_j$ and $E_j$ be the final 
$n$-qubit Pauli error
on the data qubits generated by $\calU_j$. 
Let $s_j^F\in \{0,1\}^n$ be the syndrome of the final error $E_j$.
In other words, if we write $E_j=X(\alpha_j)Z(\beta_j)$ for some vectors $\alpha_j,\beta_j\in \{0,1\}^n$, then
\[
s_j^F= 
\left[
\ba{c}
H^Z \alpha_j \\ 
H^X \beta_j\\
\ea
\right].
\]
Here $H^X$ and $H^Z$ are the check matrices of the chosen code. 
Finally, let $s_j^L\in \{0,1\}^{2k}$ be a {\em logical syndrome} of the final error $E_j$ defined as follows.
Fix some basis set of logical Pauli operators $\overline{P}_1,\overline{P}_2,\ldots,\overline{P}_{2k}$ for the chosen code.
For example, $\overline{P}_1,\overline{P}_2,\ldots,\overline{P}_k$ could be logical $X$-type operators
and  $\overline{P}_{k+1},\overline{P}_{k+2},\ldots,\overline{P}_{2k}$ could be logical $Z$-type operators.
The $i$-th bit of $s^L_j$ is defined as 
\[
(s_j^L)_i = \left\{
\ba{rcl}
1 &\mbox{if}& E_j \overline{P}_i = -\overline{P}_i  E_j,\\
0  &\mbox{if}& E_j \overline{P}_i = \overline{P}_i  E_j,\\
\ea\right.
\]
for $i=1,\ldots,2k$. Note that the pair of syndromes $s_j^F,s_j^L$ uniquely determines the final error $E_j$
modulo check operators. Define a  pair of {\em decoding matrices} $D$ and $D^L$
of size $(nN_c+n) \times M$ and $2k\times M$ respectively
such that the $j$-th column of $D$ is
\[
\left[
\ba{c}
s^U_j  \\ 
s^F_j \\
\ea
\right]
\]
and the $j$-th column of $D^L$ is $s^L_j$.
Let $p_j$ be the probability of a Pauli error that occurred  in the circuit $\calU_j$.
We have $p_j=p/15$ if $\calU_j$ contains a faulty $\cnotgate$, 
$p_j=p/3$ if $\calU_j$ contains a faulty idle qubit,
and $p_j=p$ if $\calU_j$ contains a faulty qubit initialization or measurement.
Suppose $I \subseteq \{1,2,\ldots,M\}$ is a subset of columns of 
$D$ such that  triples of syndromes  $(s^U_j,s^F_j,s^L_j)$ are the same for all
$j\in I$. We merge all columns in $I$ to a single column 
 and assign the value
  $\sum_{j\in I} p_j$ 
to the bit-flip error probability associated with the merged column.
Let $M$ be the number of columns of $D$ after the merging step
and $p_1,p_2,\ldots,p_M$ be the respective error probabilities.

Next,  the decoding matrix $D$ is converted  to a sparse form.
To this end consider a faulty circuit $\calU_j$ and a sequence of syndromes
measured by $\calU_j$ on some check operator. Let this sequence be
$m=(m_1,m_2,\ldots,m_{N_c})\in \{0,1\}^{N_c}$. Since $\calU_j$ contains a single fault, 
the sequence $m$ has only a few locations where the measured
syndromes differ at two consecutive cycles.
For example, if $\calU_j$ contains a Pauli error on some idle data qubit between two syndrome cycles,
the $m$-sequence may look as $(0,0,\ldots,0,1,1,\ldots,1)$.
Such sequence can be made sparse if we represent it by a binary vector
\[
m'=(m_1,m_2\oplus m_1,m_3\oplus m_2,\ldots, m_{N_c}\oplus m_{N_c-1}) \in \{0,1\}^{N_c}.
\]
In other words, $m'$ stores changes in the measured syndrome at a given check operators
at each cycle. We convert the matrix $D$ to a sparse form by applying the map $m\to m'$ to the syndromes measured by each check operator for each faulty circuit $\calU_j$.

Let $\xi_1,\xi_2,\ldots,\xi_M \in \{0,1\}$ be independent random variables such that
$\xi_j$ takes values $0$ and $1$ with the probability $1-p_j$ and $p_j$
respectively.
Define a linearized noise model  that outputs a random triple $(s^U,s^F,E)$,
where
\[
E=\prod_{j=1}^M (E_j)^{\xi_j}
\]
 is an $n$-qubit Pauli error
and
\[
 \left[
\ba{c}
s^U  \\ 
s^F \\
\ea
\right]= 
\sum_{j=1}^M \xi_j  \left[
\ba{c}
s^U_j  \\ 
s^F_j \\
\ea
\right] {\pmod 2}
\]
is a binary vector that represents the error syndrome.
The linearized model is a simplified version of the circuit-based noise that
ignores possible cancellations among errors generated by two or more faulty
operations in $\calU$. Note that such errors occur with the probability only $O(p^2)$.
The decoder attempts to guess the final error $E$
acting on the data qubits
based on the syndrome $s^U$ measured in the experiment  making a simplifying assumption
that that the  pair $(s^U,E)$ was generated using the linearized noise model.
We additionally assume that the decoder
knows the syndrome $s^F$ of the final error $E$.
This syndrome can be acquired  by adding one noiseless  cycle
at the end of the syndrome measurement circuit,
which is a common practice in numerical simulations of error correction.
By definition, we have
\[
D\xi= \left[
\ba{c}
s^U  \\ 
s^F \\
\ea
\right].
\]
Here $\xi=(\xi_1,\xi_2,\ldots,\xi_M)$ is  a column vector and matrix-vector multiplication is modulo two. 
Define a minimum weight error $\xi^*=\xi^*(s)\in \{0,1\}^M$ as 
a solution of an optimization problem
\be
\label{MWD}
\xi^* = \arg \min_{\xi \in \{0,1\}^M} \; \sum_{j=1}^M \log{(1/p_j)} \xi_j \quad
\mbox{subject to} \quad 
D\xi = \left[
\ba{c}
s^U  \\ 
s^F \\
\ea
\right].
\ee
This problem is equivalent to the minimum weight decoding for a length-$M$ linear code with the check matrix $D$,
 bit-flip error probabilities $p_1,p_2,\ldots,p_M$,
and noiseless syndromes. 
Our guess of the unknown logical syndrome is 
\[
s^L = D^L \xi^*.
\]
Let $E^*$ be any $n$-qubit Pauli operator
with the syndrome $s^F$ and the logical syndrome $s^L$.
Note that $E^*$ is defined uniquely modulo multiplication by check operators.
The Pauli $E^*$ is our guess of the  final error on the data qubits.
Let $E$ be the actual final error on the data qubits generated by a noisy realization of $\calU$
without making any simplifications of the noise model.
By definition, Pauli operators  $E$ and $E^*$ have the same syndrome but they may differ
by a logical Pauli operator.
We declare a logical error if $E$ and $E^*$ differ by any non-identity logical operator
(there are $4^k-1$ choices of this logical operator).
Otherwise the decoding is  deemed successful.

It remains to explain how to solve the optimization problem Eq.~(\ref{MWD}).
Since the minimum weight decoding for a linear code
 is known to be NP-hard problem~\cite{berlekamp1978inherent}, finding the exact solution
of Eq.~(\ref{MWD}) might be practically impossible for problem instances with 
several thousand variables that we
have to deal with. 
Furthermore, estimation of the logical error probability $p_L$
by the Monte Carlo method  requires solving $O(1/p_L)$
instances of the problem Eq.~(\ref{MWD}).
This number can be quite large since $p_L$ is a  small parameter. 
To address these challenges,  we employ the BP-OSD 
algorithm~\cite{panteleev2021degenerate,roffe2020decoding}.
Recall that 
belief propagation (BP)  is a heuristic message passing algorithm
aimed at  computing single-bit marginals
of a probability distribution 
\[
P(\xi|\sigma) =\left\{
\ba{rcl}
\frac1{\calZ} \prod_{j=1}^M (1-p_j)^{1-\xi_j} p_j^{\xi_j} & \mbox{if} & D\xi = \sigma,\\
0 && \mbox{otherwise}.\\
\ea
\right.
\]
Here  $\xi \in \{0,1\}^M$ and
 $\calZ$ is a normalization factor chosen such that $\sum_{\xi \in \{0,1\}^M} P(\xi|\sigma)=1$.
In our case $\xi$ represents an unknown error in the linearized
noise model, $D$ is the decoding matrix constructed above,
and $\sigma =\left[
\ba{c}
s^U  \\ 
s^F \\
\ea
\right]$ is the measured error syndrome.
Let $q_j\in [0,1]$ be an estimate of the marginal probability $\mathrm{Pr}[\xi_j=1]$
obtained by the belief propagation method with some fixed number of message
passing iterations. The 
ordered statistics
post-processing step
examines {\em information sets}
which are subsets of bits  $I \subseteq [M]$
such that the linear system $D\xi =\sigma$ has a unique solution $\xi$
supported on $I$, that is, $\xi_j=0$ for all $j\notin I$.
Information sets are ranked according to their {\em reliability}
which is defined as 
\[
\rho(I)= \prod_{j\in I} \max{(q_j,1-q_j)}.
\]
BP-OSD finds an information set $I$ with the largest reliability using a
greedy algorithm~\cite{panteleev2021degenerate}.  The final output
of BP-OSD is a 
solution of the system $D\xi=\sigma$
supported on the most reliable information set $I$.
We replace the minimum weight error $\xi^*$ in Eq.~(\ref{MWD})
by the solution $\xi$
proposed by BP-OSD.

Since BB LDPC codes are of CSS-type, it is natural to decode $X$-type and $Z$-type errors independently.
Accordingly, we solve the minimum weight decoding problem Eq.~(\ref{MWD}) twice
with a pair of decoding matrices $D_X$ and $D_Z$ constructed as above but including only the syndromes of $X$-type and $Z$-type check operators respectively. 
This results in guessed $X$-type and $Z$-type errors $E^*_X$ and $E^*_Z$.
The guessed final error is $E^*=E^*_X E^*_Z$.
We empirically observed that the resulting decoding matrices $D_X$ and $D_Z$
are $(6,35)$-sparse for any BB code, meaning that 
there are at most $6$ nonzeros in each column and at most $35$ nonzeros
in each row of $D_X$ and $D_Z$.
The number of columns  scales as $O(nN_c)$ where the constant coefficient
depends on a particular code. For example, decoding matrices $D_X$ and $D_Z$
describing the code $[[144,12,12]]$ with $N_c=12$ syndrome cycles
have $8857$ and $8785$ columns respectively.

We also employed BP-OSD to compute an upper bound on the code distance $d$.
Consider a  CSS-type LDPC code $[[n,k,d]]$ with check matrices
$H^X$ and $H^Z$.  Assume for simplicity that this code has the same
distance for $X$- and $Z$-type errors (this assumption is satisfied for BB
LDPC codes due to \ifmain{\lem{kd}}{Lemma~\lemkd in the main text}).
Suppose
$Z(\xi)$ is a minimum weight logical $Z$-type operator.
Then $\xi \in \ker{H^X}$
and $\xi \notin \rs{H^Z}$
.
Let $X(\eta)$ be any logical $X$-type operator. 
Here $\eta \in \ker{H^Z}\setminus \rs{H^X}$.
Consider the following optimization problem:
\be
\label{MWD2}
d(\eta) =  \min_{\xi \in 
\ker{H^X}} \; \sum_{j=1}^n \xi_j
 \quad
\mbox{subject to} \quad
 \eta^T \xi = 1.
\ee
Then $d(\eta)\ge d$ for any logical
operator $X(\eta)$ and $d(\eta)=d$
if $X(\eta)$ anti-commutes with some minimum-weight logical operator $Z(\xi)$.
The latter event occurs with the probability $1/2$ if one
picks 
$\eta \in \ker{H^Z}$
 uniformly at random.
In this case
 $d(\eta)=d$
with the probability at least $1/2$
and $d(\eta)\ge d$ with certainty.
Let $d^{\mathsf{BP}}(\eta)$ be an upper bound on $d(\eta)$
obtained by solving the
optimization problem Eq.~(\ref{MWD2}) using BP-OSD method
with a parity check matrix $\left[ \ba{cc}
H^X \\
\eta^T\\
\ea
\right]$
and a syndrome $(0,0,\ldots,0,1)^T$.
We have $d^{\mathsf{BP}}(\eta)\ge d$
with certainty and
 $d^{\mathsf{BP}}(\eta)=d$
  with the probability $1/2$ whenever BP-OSD finds the optimal solution. 
Choose the number trials $T\gg 1$ and pick 
vectors $\eta^1,\eta^2,\ldots,\eta^T\in \ker{H^Z}\setminus \rs{H^X}$
uniformly at random. Then 
\[
d^{\mathsf{BP}} := \min_{a=1,2,\ldots,T} d^{\mathsf{BP}}(\eta^a)
\]
is an upper bound on the distance $d$ that can be systematically improved by 
increasing the number of trials $T$. 

Using the quantity $d^{\mathsf{BP}}$ as an efficiently computable proxy for the code distance
enabled us to search over a large number of candidate BB codes
with $n=O(100)$ qubits. The vast majority of these candidates was discarded due to an
insufficiently large upper bound $d^{\mathsf{BP}}$. 
This left only a few  viable candidates with a sufficiently large value of $d^{\mathsf{BP}}$.
The actual distance of each  candidate code was computed 
using the  integer linear programming method~\cite{landahl2011fault}.

We similarly 
computed an upper bound on the circuit-level distance $\dcirc$. 
Since the SM circuit can break the symmetry between $X$- and $Z$-type errors, the circuit-level distance has to be computed for both types of errors. 
For concreteness, let us discuss
the circuit-level distance 
$\dcirc^Z$ for $Z$-type errors.
The latter is defined as the minimum number of faulty operations in the SM circuit that can generate an undetectable $Z$-type logical error. 
The optimization problem 
Eq.~(\ref{MWD2}) is replaced by
\be
\label{MWD1}
\dcirc^Z(\eta) =  \min_{\xi \in
\ker{D_X}}
\; \sum_{j=1}^M \xi_j \quad
\mbox{subject to} \quad
\eta^T \xi = 1,
\ee
where $D_X$ is the decoding matrix constructed above
and $\eta\in \{0,1\}^M$
is a random linear combination of rows of $D_X$ and rows of $D_L$ that represent logical $X$-type operators. 
Then $\dcirc(\eta)\ge \dcirc^Z$
with certainty and $\dcirc(\eta)=\dcirc^Z$
with the probability at least $1/2$.
Solving the optimization problem
Eq.~(\ref{MWD1}) using BP-OSD method for many  random choices of the vector $\eta$ and taking the minimum value of $\dcirc^Z(\eta)$
provides an upper bound
on $\dcirc^Z$.
One can similarly compute an upper
bound on the circuit-level distance $\dcirc^X$ for $X$-type errors. This provides an upper bound on $\dcirc=\min{(\dcirc^X,\dcirc^Z)}$.

\section{Proof of Lemma 1}

For convenience of the reader we restate the lemma below.
\setcounter{lemma}{0}
\begin{lemma}
The code $\qc$ has parameters $[[n,k,d]]$, where
\[
n = 2\ell m,\quad 
k = 2\cdot\mathrm{dim}\left(\ker{A} \cap \ker{B}\right),
\quad \mbox{and} \quad
d= \min\bigl\{ |v|{:} \,\, v\in \ker{H^X}{\setminus} \rs{H^Z} \bigr\}.
\]
The code offers equal distance for $X$-type and $Z$-type errors.
\end{lemma}
\begin{proof}
It is known~\cite{steane1996multiple,calderbank1996good} that 
\[
k = n - \rk{H^X}-\rk{H^Z}.
\]
We claim that $\rk{H^X}\,{=}\,\rk{H^Z}$.
Indeed, define a self-inverse permutation matrix $C_\ell$ of size $\ell \,{\times}\, \ell$ such that the $i$-th column of $C_\ell$ has a single nonzero entry equal to one at the row $j=-i{\pmod \ell}$. Define $C_m$ similarly
and let $C=C_\ell \otimes C_m$. Since $C_\ell S_\ell C_\ell = S_\ell^T$
and $C_m S_m C_m = S_m^T$, one gets
\be
\label{ABC2}
A^T = C AC \quad \mbox{and} \quad B^T = CBC.
\ee
Therefore one can write
\[
H^Z= [B^T|A^T] = [CBC |CAC] = C[A|B]\left[ \ba{cc} 0 & C  \\ C & 0\\ \ea \right]
=CH^X\left[ \ba{cc} 0 & C  \\ C & 0\\ \ea \right].
\]
Thus $H^Z$ is obtained from $H^X$ by multiplying on the left and on the right by invertible matrices.
This implies $\rk{H^X}=\rk{H^Z}$. Therefore
\begin{align*}
k& = n-2{\cdot}\rk{H^Z}=n-2{\left(\frac{n}2-\dim{(\ker{(H^Z)^T)})}\right)} = n - 2{\left(\frac{n}2-\dim{(\ker{A}\cap \ker{B})}\right)}\\
&= 2\cdot\mathrm{dim}\left(\ker{A} \cap \ker{B}\right).
\end{align*}
Here we noted that $H^Z$ has size $(n/2)\,{\times}\, n$ and $\ker{(H^Z)^T)}=\ker{A}\cap \ker{B}$ since $H^Z=[B^T|A^T]$.

It is known~\cite{steane1996multiple,calderbank1996good}  that 
a CSS code with check matrices $H^X$ and $H^Z$ has distance
$d\,{=}\,\min{(d^X,d^Z)}$, where 
$d^X$ and $d^Z$ are the code distances for $X$-type and $Z$-type errors
defined as
\[
d^X =  \min\bigl\{ |v|{:} \,\, v\in \ker{H^Z}{\setminus} \rs{H^X} \bigr\}
\quad \mbox{and} \quad
d^Z =  \min\bigl\{ |v|{:} \,\, v\in \ker{H^X}{\setminus} \rs{H^Z} \bigr\}.
\]
We claim that $d^Z\,{\le}\, d^X$.
Indeed, let $X(f)\,{=}\,\prod_{j=1}^n X_j^{f_j}$ be a minimum weight logical $X$-type Pauli operator
such that $|f|\,{=}\,d^X$. 
Then $H^Z f\,{=}\,0$ and $f\,{\notin}\, \rs{H^X}$.
Thus there exists a logical $Z$-type operator $Z(g)\,{=}\,\prod_{j=1}^n Z_j^{g_j}$ anti-commuting with $X(f)$. In other words, $H^X g\,{=}\,0$ and $f^T g \,{=}\, 1$. Here, $f$ and $g$ are length-$n$ binary vectors.
Write $f\,{=}\,(\alpha,\beta)$ and $g\,{=}\,(\gamma, \delta)$, where $\alpha,\beta,\gamma,\delta$ are length-$(n/2)$ vectors.
Conditions $H^Z f =0$ and $H^X g=0$ are equivalent to 
\be
\label{alpha_beta_gamma_delta}
B^T \alpha = A^T \beta \quad \mbox{and} \quad A \gamma = B \delta.
\ee
Here and below all arithmetics is modulo two.
Define length-$n$ vectors 
\be
e = (C\delta, C\gamma) \quad \mbox{and} \quad h=(C \beta, C\alpha).
\ee
From Eqs.~(\ref{ABC2},\ref{alpha_beta_gamma_delta}) one gets
\[
H^X h = [A|B]\left[\ba{c} C \beta \\ C \alpha\\ \ea\right]
=AC\beta+BC\alpha = C(A^T \beta + B^T \alpha) = 0.
\]
Likewise, 
\[
H^Z e = [B^T|A^T]\left[\ba{c} C \delta \\ C \gamma\\ \ea\right]
=B^T C\delta + A^T C\gamma = C(B \delta + A \gamma) = 0.
\]
Furthermore, 
\[
h^T e = \beta^T C C \delta + \alpha^T C C\gamma = \beta^T \delta + \alpha^T \gamma =f^T g=1.
\]
Thus $X(e)$ and $Z(h)$ are non-identity logical operators. 
It follows that $d^Z\le |h|$. We get
\[
d^Z\le |h|=|C\beta|+|C\alpha|=|\beta|+|\alpha|=|f|=d^X.
\]
Thus $d^Z\le d^X$. Similar argument shows that $d^X\le d^Z$, that is, $d^X=d^Z$.
\end{proof}
\noindent
We note that
the equality $d^X=d^Z$ can also be established using the machinery of
Ref.~\cite{panteleev2021quantum} 
by viewing $\qc$ as a Lifted Product code.

\section{Numerical simulation details}\label{sec:numerics}


Data reported in \ifmain{\fig{numerics} A)}{Figure~\figsmallldpc in the main text} was generated using BP-OSD software developed by Roffe et al.~\cite{roffe2020decoding,Roffe_LDPC_Python_tools_2022}.
The decoder was extended to the circuit-based noise model as described in \sec{decoder}.
The simulations were performed  using MIN-SUM belief propagation with the limit of $10,000$ iterations and 
 combination sweep version of OSD, as described in~\cite{roffe2020decoding}.
All data points except for those with the smallest error rate accumulated at least 100 logical errors to estimate the logical error rate  $p_L$ with the error bars $\approx p_L/10$.
The fitting formula
$p_L(p)=p^{\dcirc/2} e^{c_0 + c_1 p + c_2 p^2}$ 
with fitting parameters $c_0,c_1,c_2$
was proposed in~\cite{bravyi2013simulation} in the context of surface code simulations. We observed that the same fitting formula 
works well for BB LDPC codes. The fitting parameters $c_i$ of the considered codes are provided in \tab{fit}.
\edit{We note that the logical error rate achieved by the combination of a distance-preserving SM circuit
and an optimal decoder is expected to follow an exponential decay $p_L(p)=\exp[-d\cdot f(p)]$, where
$f(p)$ is an unknown function such that $f(p)\,{>}\,0$ in the sub-threshold regime.
The function $f(p)$ must have a logarithmic singularity $f(p) \approx -(1/2)\log{p}$ for small $p$ since one expects degree-$(d/2)$ error suppression for a distance-$d$ code. The fitting formula for $p_L(p)$ approximates the remaining non-singular terms in $f(p)$ by a low-degree polynomial in $p$. Coefficients of the polynomial are considered as fitting parameters. 
Since our SM circuit is not distance-preserving, the code distance $d$ in the fitting formula of Ref.~\cite{bravyi2013simulation} is replaced by the circuit-level distance $\dcirc$.}


Surface code data reported in \ifmain{\fig{numerics}~B)}{Figure~\figldpcvssurface in the main text} was generated using software developed by one of the authors and Alexander Vargo in~\cite{bravyi2013simulation}. The simulation was performed for the rotated surface code with parameters $[[d^2,1,d]]$,
where  
\edit{$d\in\{9,11,13,15\}$},
and the standard SM
 circuit~\cite{fowler2009high}.
Let $P_{L,1}$ be the logical error probability 
for the surface code encoding one logical qubit and SM circuit with $N_c=d$ syndrome cycles. 
Encoding $k=12$ logical qubits into $12$ separate patches of the surface code results in a logical error probability
\[
P_{L,12} =
 1-(1-P_{L,1})^{12}.
\]
\ifmain{\fig{numerics}~B)}{Figure~\figldpcvssurface in the main text} shows the logical error rate $p_L$ defined as the logical error probability per syndrome cycle,
\[
p_L = 1-(1-P_{L,12})^{1/N_c} = 1-(1-P_{L,1})^{12/d}.
\]

\begin{table}[t]
\begin{center}
\begin{tabular}{c|c|c|c}
\hline
$[[n,k,d]]$ & $c_0$ & $c_1$ & $c_2$ \\
\hline
\hline
$[[72,12,6]]$ & $11.09$  & $365.6$ &  $-16088$  \\
\hline 
$[[90,8,10]]$ & $15.08$  & $524.8$ &  $-12670$  \\
\hline
$[[108,8,10]]$ & $13.91$  & $895$ & $-46137$    \\
\hline
$[[144,12,12]]$ & $18.04$  &  $1337$ & $-96007$    \\
\hline
$[[288,12,18]]$ & $32.04$  & $3522$ & $-294482$    \\
\end{tabular}
\end{center}
\caption{Parameters
$c_0,c_1,c_2$ in the fitting formula 
$p_L(p)=p^{\dcirc/2} e^{c_0 + c_1 p + c_2 p^2}$ 
for BB LDPC codes shown in \ifmain{\tab{codes_intro}}{Table~\tabcodesintro in the main text}.}
\label{table:fit}
\end{table}

\section{Logical memory capabilities}\label{sec:logicals}

In this section we give evidence that BB LDPC codes have the required features for an effective quantum memory or storage unit. Although there are few ways of performing computations on stored qubits, there are fault tolerant operations for initialization and measurement of individual qubits, and most importantly transfer of data into and out of the code via quantum teleportation. These capabilities are based on a combination of two different techniques.
First, we follow \cite{breuckmann2022foldtransversal} to derive fault tolerant unitary operations that require only the connectivity already necessary to perform syndrome measurements. Second, we give low-overhead extensions of the Tanner graph based on work by  \cite{cohen2022lowoverhead} which enable measurement of a single logical operator while preserving the thickness-2 implementability criterion. Together, these capabilities allow us to address any logical qubit.

\begin{figure}[b!]
  \centering
    \includegraphics[height=7cm]{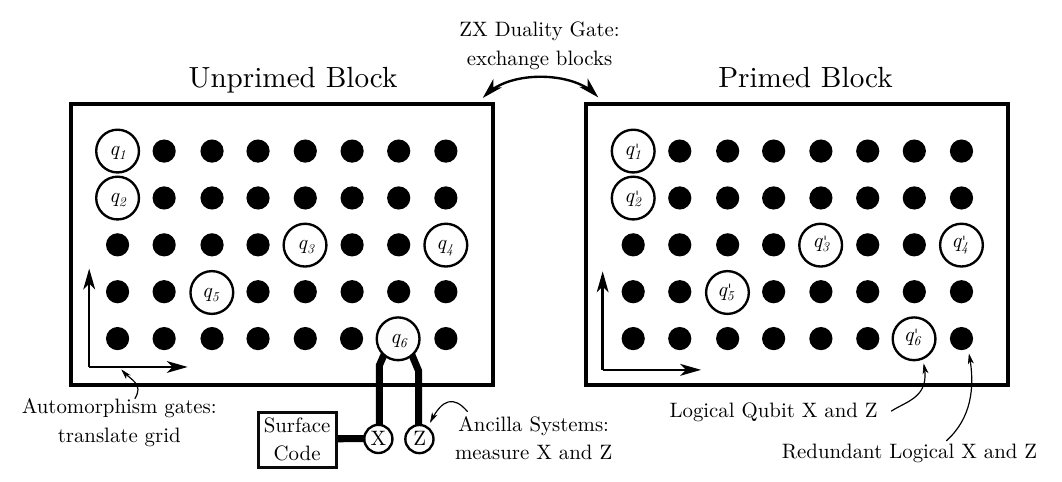}
     \caption{Conceptual diagram depicting the manner by which logical operators can be loaded into and out of a BB LDPC code. In \ssec{logical_pauli_operators} we derive that there are two blocks of logical Pauli operators corresponding to a 2d grid. Some subset of these grid elements can be chosen as logical qubits (large dots) and the other elements correspond to various Pauli products (small dots). In \ssec{automorphisms} we show that there are fault tolerant logical gates based on automorphisms that translate the grid of operators within each block, and in \ssec{ZX_duality} we give a fault tolerant gate based on a ZX-duality \cite{breuckmann2022foldtransversal} that swaps the two blocks. Finally, in \ssec{logical_measurements} we show that there exists an ancilla system based on \cite{cohen2022lowoverhead} that can measure one logical $X$ operator. We can think of this system as a probe that can access one logical qubit. Together, these operations allow external access of every logical qubit and many of their products.} 
	  
  \label{fig:ops_summary}
\end{figure}

A conceptual representation of the logical operators is shown in \fig{ops_summary}. We first derive logical Pauli operators for BB LDPC codes, and find that the logical qubits divide into an ``unprimed'' and a ``primed'' block with equal size and symmetrical structure. We visualize the primed and unprimed block as two sheets featuring a 2D grid of logical operators. Some of these grid cells contain one of the $k/2$ logical qubits per sheet.

Next, we show that there exists a set of fault tolerant depth-four circuits that implement a small family of commuting logical $\cnotgate$ circuits. These gates are based on automorphisms of the code: permutations of the data qubits that commute with the stabilizer. Based on their group structure we can think of the automorphism gates as translations of a 2D grid of operators within each of the primed and unprimed blocks. Furthermore, we follow  \cite{breuckmann2022foldtransversal} to derive a fault tolerant operation based on a ZX-duality that also allows us to swap the two blocks while also applying Hadamard gates to all qubits.

Finally, we show how to leverage techniques from \cite{cohen2022lowoverhead} to extend the Tanner graph to a larger Tanner graph allowing fault-tolerant measurement of one logical $X$ and one logical $Z$ operator. Various subgraphs of this extended Tanner graph contain either this logical $X$ or $Z$ operator as a stabilizer. This construction acts as a ``probe'' that gives us access to one of the logical qubits. 

Measurements of both logical $X$ and $Z$ operators on any qubit can be realized by conjugating this measurement by gates based on automorphisms and the ZX-duality. We can think of this as shifting any desired qubit to be the target of the probe using translation and exchange of the two blocks.

Logical $X$ and $Z$ measurement of any logical qubit also enables transfer of data into and out of the code using a teleportation circuit. This teleportation can be realized through a product measurement of the logical Pauli in the BB code, and a logical Pauli in another quantum error correction code. While the Tanner graph of the BB code demands thickness-2, we show how the ancilla system corresponding to the logical $X$ measurement can be implemented in an ``effectively planar'' Tanner graph. This makes it possible two connect this ancilla system to another quantum error correction code, like a surface code, within a thickness-2 implementation. This capability indicates the suitability of BB LDPC codes as a fault tolerant quantum memory.

\edit{On the other hand, this construction incurs rather significant resource overhead that undercuts the compactness of the error correction codes introduced in this paper. For example, to equip the $[[144,12,12]]$ code with ancilla systems capable of measuring $X$ and $Z$, we require $2 \,{\times}\, 30 \,{\times}\, (2d{-}1) = 1380$ additional qubits on top of the original 288. However, the argument for fault-tolerance from \cite{cohen2022lowoverhead} is designed to be very broadly applicable, and hence may demand excessively many resources for any particular error correction code. We consider it very likely that the size of these systems can be significantly reduced. We leave resource optimization of this scheme for future work.}

\subsection{Logical Pauli Operators}
\label{ssec:logical_pauli_operators}

In this section we derive that the logical Pauli matrices of BB LDPC codes split into a ``primed'' and an ``unprimed'' block with $|\mathcal{M}| = \ell m$ many $X$ operators and $Z$ operators each. Operators in the primed block commute with operators in the unprimed block, and the two blocks have identical commutation structure.

We begin by introducing some new notation for Pauli matrices acting on the data qubits. We denote with $\mathbb{F}_2^{\mathcal{M}}$ the set of polynomials over $\mathbb{F}_2$ with monomials from $\mathcal{M}$. This is equivalent to the quotient ring obtained from $\mathbb{F}_2[x,y]$ by identifying $x^\ell= y^m = 1$. With $x = S_m\otimes I$ and $y = I \otimes S_\ell$, the elements of $\mathbb{F}_2^{\mathcal{M}}$ have natural matrix representations, and can also be interpreted as sets since the coefficient on any particular $x \in \mathcal{M}$ is either 0 or 1.

For $P,Q \in \mathbb{F}_2^{\mathcal{M}}$, we can consider the set of qubits $q(L,\alpha)$ for $\alpha \in P$ and $q(R,\beta)$ for $\beta \in Q$.   We write $X(P,Q)$ to denote a Pauli matrix acting as $X$ on this collection of qubits, and identity elsewhere. Similarly, $Z(P,Q)$ denotes $Z$ acting on $q(L,\alpha)$ for $\alpha \in P$ and $q(R,\beta)$ for $\beta \in Q$. For example, we can recall that $q(L,\beta)$ is connected to $q(X,\alpha)$ whenever $\beta \in A\alpha$, and see that the stabilizer corresponding to $q(X,\alpha)$ becomes $X(\alpha A, \alpha B)$. Similarly, the stabilizer corresponding to $q(Z,\alpha)$ can be written as $Z(\alpha B^T,\alpha A^T)$. There is also the following useful fact:

\begin{lemma} $X(P,Q)$ anticommutes with $Z(\bar P,\bar Q)$ if and only if $1 \in P\bar P^T + Q\bar Q^T$.
\end{lemma}
\begin{proof}
Write 
\[
P=\sum_{\alpha\in \calM} p_\alpha \alpha,
\quad 
\bar{P}=\sum_{\alpha\in \calM} \bar{p}_\alpha \alpha,
\quad
Q=\sum_{\alpha\in \calM} q_\alpha \alpha,
\quad
\bar{Q}=\sum_{\alpha\in \calM} \bar{q}_\alpha \alpha,
\]
where  $p_\alpha,\bar{p}_\alpha,q_\alpha,\bar{q}_\alpha\in \FF_2$ are coefficients.
Pauli operators $X(P,Q)$ and $Z(\bar{P},\bar{Q})$ overlap on a qubit
$q(L,\alpha)$  iff $p_\alpha \bar{p}_\alpha=1$
and overlap on a qubit $q(R,\alpha)$ iff $q_\alpha\bar{q}_\alpha=1$.
Thus  $X(P,Q)$ and $Z(\bar{P},\bar{Q})$ anti-commute
iff  $\sum_{\alpha \in \calM} p_\alpha \bar{p}_\alpha + q_\alpha \bar{q}_\alpha$
is odd. 
We have 
\[
P\bar{P}^T = \sum_{\alpha \in \calM} p_\alpha \bar{p}_\alpha 1 + \ldots
\quad \mbox{and} \quad
Q\bar{Q}^T =  \sum_{\alpha \in \calM} q_\alpha \bar{q}_\alpha 1 + \ldots
\]
where dots represent all monomials different from $1$. By linearity,
\[
P\bar{P}^T + Q\bar{Q}^T
=\sum_{\alpha \in \calM}
(p_\alpha \bar{p}_\alpha
+q_\alpha \bar{q}_\alpha)1+\ldots.
\]
Thus $X(P,Q)$ and $Z(\bar{P},\bar{Q})$ anti-commute
iff $P\bar{P}^T + Q\bar{Q}^T$ contains the monomial $1$.
\end{proof}

Without loss of generality, we can express logical Pauli matrices as either $X(P,Q)$ or $Z(Q^T, P^T)$ via a choice of $P,Q \in \mathbb{F}_2^{\mathcal{M}}$. The operator $X(P,Q)$ commutes with the stabilizer $Z(\alpha B^T, \alpha A^T)$ whenever $1 \not\in P(\alpha B^T)^T + Q (\alpha A^T)^T = \alpha^T (PB + QA) $. This is equivalent to $\alpha \not\in PB + QA$. Since we must have $\alpha \not\in PB + QA$ for all $\alpha$, we see that $X(P,Q)$ commutes with the stabilizer whenever $PB + QA$ vanishes. Similarly we can derive that $Z(Q^T, P^T)$ commutes with the stabilizer when $PB + QA = 0$. 

We aim to construct a family of solutions to $PB + QA = 0$ which give rise to a basis of logical qubits defined by a set of operators $\{\bar X_1, \bar X_2, ..., \bar Z_1, \bar Z_2, ... \}$ with the correct commutation relations.  To do so, let us make some observations about Pauli operators defined via solutions to $PB + QA = 0$. First, if $P,Q$ are a solution, then so are $\alpha P, \alpha Q$ for any $\alpha \in \mathcal{M}$, so each $P,Q$ immediately gives rise to a family of $|\mathcal{M}| = lm$ logical operators for both $X$ and $Z$. Second, consider using the same $P,Q$ to define both $X(\alpha P, \alpha Q)$  and $Z(\beta Q^T, \beta P^T)$. Then these operators always commute because $\beta \alpha^T \in PQ + QP = 0$ never holds. So we require at least two solutions to $PB+QA = 0$ to define a set of operators with nontrivial commutation relations. 

For reasons described later in \ssec{logical_measurements}, we would like a logical $X$ operator with no support on $q(R)$. To this end, we select $f,g,h \in \mathbb{F}_2^{\mathcal{M}}$ that satisfy $Bf = 0$ and $gB + hA = 0$, yielding two solutions to the equation $PB+QA = 0$ with $P,Q = f,0$ and $P,Q = g,h$. These yield the following family of logical operators for all $\alpha \in \mathcal{M}$:
\begin{equation}
\begin{aligned}
	 \bar X_\alpha &:= X(\alpha f,0) \hspace{1.2cm} \bar Z_\alpha := Z(\alpha h^T, \alpha g^T)\\
    \bar X'_\alpha &:= X(\alpha  g, \alpha  h) \hspace{1cm} \bar Z'_\alpha := Z(0, \alpha f^T) 
\end{aligned}\label{eq:logicalpaulis}
\end{equation}

For all $\alpha,\beta $, we see that $\bar X_\alpha, \bar Z'_\beta$ always commute because $f0^T + 0f^T = 0$, and $\bar X'_\beta,\bar Z_\alpha$ always commute because $gh+hg=0$. Furthermore, $\bar X_\alpha, \bar Z_\beta$ and $\bar X'_\alpha, \bar Z'_\beta$ form anticommuting pairs when $\alpha^T\beta \in fh$. We see that we have constructed two independent blocks of logical operators with symmetrical structure. It follows that each of these blocks must contain a set of operators that define $k/2$ qubits. We name these the ``unprimed'' and ``primed'' logical blocks with $\bar X_\alpha, \bar Z_\beta$ and $\bar X'_\alpha, \bar Z'_\beta$ respectively.

Not all choices of $f,g,h$ span all $k$ logical qubits, but valid choices are readily enumerated in software. Solutions to $Bf = 0$ and $gB + hA = 0$ correspond to null spaces of $B$ and $\begin{bmatrix}B \\ A\end{bmatrix}$ respectively, which can be constructed by Gaussian elimination. Gaussian elimination can also be used to check if the operators $\bar X_\alpha, \bar Z_\alpha, \bar X'_\alpha, \bar Z'_\alpha$ together span $k$ qubits up to the stabilizer. We find all codes in \ifmain{\tab{codes}}{Table~\tabcodes in the main text} admit several such choices of $f,g,h$. In \tab{f_polys} we show a particularly favorable choice of $f,g,h$ for the $[[144,12,12]]$ code where the resulting logical operators have minimum weight.


To identify logical qubits we can enumerate choices of monomials $\{n_1, n_2, \ldots,n_{k/2}\}$ and $\{m_1, m_2, \ldots,m_{k/2}\}$ such that $n_i^T m_j \in fh$ exactly when $i = j$.  That way, $\bar X_{n_i}, \bar Z_{m_i}$ as well as $\bar X'_{n_i}, \bar Z'_{m_i}$  for $i = 1...k/2$ form a set of $k$ logical qubits: $\bar X_{n_i}$ anticommutes with $\bar Z_{m_j}$ exactly when $i = j$. A brute force search readily finds choices of $\{n_i\},\{m_i\}$.

\begin{table}
	\begin{center}
\begin{tabular}{|c|c|c|}
	\multicolumn{3}{c}{Polynomials for Logical Pauli Matrices in the $[[144,12,12]]$ Code}\\[2mm]\hline
	$f$ & $g$ & $h$ \\ \hline
	\resizebox{0.3\textwidth}{!}{\begin{tikzpicture}
			\draw (-0.4,-0.75) node[anchor=south west] {$x$};
			\draw (-0.75,-0.45) node[anchor=south west] {$y$};
			\foreach \x in { -1,0,1,2,3,4,5,6,7,8,9,10,11 }
			\draw (\x / 2 + 0.5,-0.5) -- (\x / 2 + 0.5,3.0);
			\foreach \x in { 0,1,2,3,4,5,6,7,8,9,10,11 }
			\draw (\x / 2 + 0.25,-0.5) node[anchor=south] { \x };
			\foreach \y in { -1,0,1,2,3,4,5 }
			\draw (-0.5, \y / 2 + 0.5) -- (6.0, \y / 2 + 0.5);
			\foreach \y in { 0,1,2,3,4,5 }
			\draw (-0.5, \y / 2 + 0.25) node[anchor=west] { \y };
			\foreach \x / \y in { 1 / 0,0 / 0,2 / 0,3 / 0,6 / 0,7 / 0,8 / 0,9 / 0,1 / 3,5 / 3,7 / 3,11 / 3 }
			\draw[fill=black] (\x / 2 + 0.25,\y / 2 + 0.25) circle (0.15);
			\draw[transparent] (-0.5, 3.25) -- (6.0,3.25);
	\end{tikzpicture}}  & \resizebox{0.3\textwidth}{!}{\begin{tikzpicture}
			\draw (-0.4,-0.75) node[anchor=south west] {$x$};
			\draw (-0.75,-0.45) node[anchor=south west] {$y$};
			\foreach \x in { -1,0,1,2,3,4,5,6,7,8,9,10,11 }
			\draw (\x / 2 + 0.5,-0.5) -- (\x / 2 + 0.5,3.0);
			\foreach \x in { 0,1,2,3,4,5,6,7,8,9,10,11 }
			\draw (\x / 2 + 0.25,-0.5) node[anchor=south] { \x };
			\foreach \y in { -1,0,1,2,3,4,5 }
			\draw (-0.5, \y / 2 + 0.5) -- (6.0, \y / 2 + 0.5);
			\foreach \y in { 0,1,2,3,4,5 }
			\draw (-0.5, \y / 2 + 0.25) node[anchor=west] { \y };
			\foreach \x / \y in { 1 / 2,0 / 4,2 / 3,1 / 0,0 / 2,2 / 1 }
			\draw[fill=black] (\x / 2 + 0.25,\y / 2 + 0.25) circle (0.15);
			\draw[transparent] (-0.5, 3.25) -- (6.0,3.25);
	\end{tikzpicture}}  & \resizebox{0.3\textwidth}{!}{\begin{tikzpicture}
			\draw (-0.4,-0.75) node[anchor=south west] {$x$};
			\draw (-0.75,-0.45) node[anchor=south west] {$y$};
			\foreach \x in { -1,0,1,2,3,4,5,6,7,8,9,10,11 }
			\draw (\x / 2 + 0.5,-0.5) -- (\x / 2 + 0.5,3.0);
			\foreach \x in { 0,1,2,3,4,5,6,7,8,9,10,11 }
			\draw (\x / 2 + 0.25,-0.5) node[anchor=south] { \x };
			\foreach \y in { -1,0,1,2,3,4,5 }
			\draw (-0.5, \y / 2 + 0.5) -- (6.0, \y / 2 + 0.5);
			\foreach \y in { 0,1,2,3,4,5 }
			\draw (-0.5, \y / 2 + 0.25) node[anchor=west] { \y };
			\foreach \x / \y in { 0 / 3,0 / 2,1 / 3,0 / 1,0 / 0,1 / 1 }
			\draw[fill=black] (\x / 2 + 0.25,\y / 2 + 0.25) circle (0.15);
			\draw[transparent] (-0.5, 3.25) -- (6.0,3.25);
	\end{tikzpicture}}   \\ \hline
	\multicolumn{3}{c}{}\\
\end{tabular}
	\end{center}
	\caption{\label{table:f_polys} Choices of polynomials $f,g,h$ such that $\bar X_\alpha := X(\alpha f,0)$ and $\bar Z_\alpha = Z(\alpha h^T, \alpha g^T)$ as defined in equation \ref{eq:logicalpaulis} are minimum-weight logical Pauli operators. The dots represent the monomials of the form $x^iy^j$ with coefficient 1.  If we let $\{n_i\} = \{1, y, x^2y, x^2y^5, x^3y^2, x^4\}$ and $\{m_i\} = \{ y, y^5, xy, 1, x^4, x^5y^2 \}$, then $\bar X_{n_i},\bar Z_{m_j}$ anticommute exactly when $i = j$. When used to construct an ancilla system as in Section~\ref{ssec:logical_measurements}, these polynomials give a system with 60 qubits per layer.
	} 
\end{table}

\subsection{Logical Gates based on Automorphisms}
\label{ssec:automorphisms}

An automorphism of an error correction code is a permutation of the physical qubits that is equivalent to a permutation of the checks 
(more generally, an automorphism can map a check operator to a product of check operators).  
We focus on permutations that are implementable using fault tolerant circuits within the connectivity already required for syndrome measurements.

The existing connectivity admits some natural fault tolerant circuits implementing a particular family of permutations on the data qubits. BB LDPC codes feature two data registers $q(L),q(R)$ and two
check registers $q(X),q(Z)$. We consider circuits that transfer the qubits from the data registers to the ancilla registers, and back again on a different path. The adjacency matrices describing the connectivity between the data and the ancilla registers are given by $A$ and $B$, which are the sum of three monomials $A_1, A_2, A_3$ and $B_1,B_2,B_3$ in $\mathcal{M}$. Each monomial is a permutation and thus describes a vertex-disjoint set of edges between the data and ancilla block. Hence, all swaps along these edges can be parallelized.
In a single circuit we can either swap along the edges defined by $A$ which are $q(L)\leftrightarrow q(X)$ and $q(R)\leftrightarrow q(Z)$, or along edges defined by $B$ which are $q(L)\leftrightarrow q(Z)$ and $q(R)\leftrightarrow q(X)$. See also \ifmain{\fig{navigation} A)}{Figure~\fignavigation A) in the main text}. 

The monomial defining the particular set of edges in each of these sets of swaps can be chosen independently for each stage of the permutation (data $\to$ ancilla or ancilla $\to$ data), and on each side of the Tanner graph. 
For example, we can select any $A_j, A_k, A_{j'}, A_{k'}$ and move $q(L) \to_{A_j^T} q(X) \to_{A_{k}} q(L)$ and simultaneously move $q(R) \to_{A_{k'}} q(Z) \to_{A_{j'}^T} q(R)$. However, we will see later that it is necessary to select $A_j = A_{j'}$ and $A_k = A_{k'}$. Furthermore, these swaps admit a standard optimization: if we initialize the check registers $q(X),q(Z)$ to the $|0\rangle$ state, then circuits implementing these permutations have $\cnotgate$ depth four. If we also reset qubits to the $|0\rangle$ state in between the swaps wherever possible, we obtain circuits whose errors cannot propagate between physical qubits, and are hence fault tolerant. See \tab{automorphisms}.

We now verify that the permutations implemented by the circuits described above are indeed automorphisms. After having applied an `A' type permutation based on $A_j,A_k$, the qubits are permuted by $q(L, \alpha) \leftrightarrow q(L,A^T_k A_j \alpha )$ and $q(R, \alpha) \leftrightarrow q(R,A^T_{k}A_{j}\alpha )$. We see that this transforms a Pauli matrix by $X(P,Q) \to X( A_j A_k^T P, A_j A_k^T Q  )$. Consequently, the stabilizers are transformed as $X(\alpha A, \alpha B) \to X(\alpha  A_j A_k^T A, \alpha A_j A_k^T B  )$, which is the same as permuting the $X$ checks by $\alpha \to \alpha A_j A_k^T$. The $Z$ stabilizers are also permuted by $\alpha \to \alpha A_j A_k^T$, so the described circuit indeed implements an automorphism. Notice also that this only works because the $q(L)$ and $q(R)$ blocks were transformed by the same $A_j A_k^T$. The `B' type permutations can be verified to be automorphisms in the same manner, permuting the checks by some $B_jB_k^T$.

These automorphisms allow us to fault tolerantly implement a subgroup of the Clifford gates. As we saw in \ifmain{\lem{connected}}{Lemma~\lemconnected in the main text}, shifts of the form $A_j A_k^T$ or $B_jB_k^T$ generate the entire group $\mathcal{M}$ whenever the Tanner graph is connected.
Therefore, by leveraging these permutations as generators, we can perform all translations of the tori containing $q(L), q(R)$ using fault tolerant circuits of varying depth. An automorphism defined by an element $s \in \mathcal{M}$ transforms $\bar X_{n_i} \to \bar X_{sn_i}$, $\bar Z_{m_i} \to \bar Z_{sm_i}$ and similarly for the primed logical Pauli matrices.  This capability is critical for addressing all logical qubits.

We can also comment on the nature of these operations as logical gates, although they are less useful in this sense. There is one such operation per element in $\mathcal{M}$, and since $\mathcal{M}$ is Abelian the subgroup of Clifford gates implemented by these automorphisms must be Abelian as well.  A transformation of this form cannot act like the logical identity so all of these gates (except $s= 1$) are nontrivial. Since automorphism operations take $\bar X$ to $\bar X$ and $\bar Z$ to $\bar Z$, and they must hence be logical $\cnotgate$ circuits up to a logical Pauli correction. While it is not clear how to use these $\cnotgate$ circuits to facilitate useful computations, they may make for interesting test cases in an implementation.

\begin{table}[t]
\begin{center}
\begin{tabular}{cc}
\multicolumn{2}{c}{`A' type automorphism based on any $A_j,A_k$}\\
{\centering
\begin{minipage}{7cm}
	\begin{algorithmic}
	\State{}
	\For{$\alpha \in \mathcal{M}$}
	\State{$\initZ{q(X,\alpha)}$}
	\State{$\cnot{q(L,A_j\alpha )}{q(X,\alpha)}$}
	\State{$\cnot{q(X,\alpha)}{q(L,A_j\alpha)}$}
	\State{$\initZ{q(L,A_k^T\alpha)}$}
	\State{$\cnot{q(X,\alpha)}{q(L,A_k\alpha)}$}
	\State{$\cnot{q(L,A_k\alpha)}{q(X,\alpha)}$}
	\EndFor
	\State{}
	\end{algorithmic}
\end{minipage}}
&
{\centering
\begin{minipage}{7cm}
	\begin{algorithmic}
	\State{}
	\For{$\alpha \in \mathcal{M}$}
	\State{$\initZ{q(Z,\alpha)}$}
	\State{$\cnot{q(R,\alpha)}{q(Z,A_{j}\alpha)}$}
	\State{$\cnot{q(Z,A_{j}\alpha)}{q(R,\alpha)}$}
	\State{$\initZ{q(R,A_k\alpha)}$}
	\State{$\cnot{q(Z,A_{k}\alpha)}{q(R,\alpha)}$}
	\State{$\cnot{q(R,\alpha)}{q(Z,A_{k}\alpha)}$}
	\EndFor
	\State{}
	\end{algorithmic}
\end{minipage}}
 \\ \hline \\
 \multicolumn{2}{c}{`B' type automorphism based on any $B_j,B_k$}\\
{\centering
\begin{minipage}{7cm}
	\begin{algorithmic}
	\State{}
	\For{$\alpha \in \mathcal{M}$}
	\State{$\initZ{q(X,\alpha)}$}
	\State{$\cnot{q(R,B_j\alpha)}{q(X,\alpha)}$}
	\State{$\cnot{q(X,\alpha)}{q(R,B_j\alpha)}$}
	\State{$\initZ{q(R,B^T_k\alpha)}$}
	\State{$\cnot{q(X,\alpha)}{q(R,B_k\alpha)}$}
	\State{$\cnot{q(R,B_k\alpha)}{q(X,\alpha)}$}
	\EndFor
	\State{}
	\end{algorithmic}
\end{minipage}}
&
{\centering
\begin{minipage}{7cm}
	\begin{algorithmic}
	\State{}
	\For{$\alpha \in \mathcal{M}$}
	\State{$\initZ{q(Z,\alpha)}$}
	\State{$\cnot{q(L,\alpha)}{q(Z,B_{j}\alpha)}$}
	\State{$\cnot{q(Z,B_{j}\alpha)}{q(L,\alpha)}$}
	\State{$\initZ{q(L,B_k\alpha)}$}
	\State{$\cnot{q(Z,B_{k}\alpha)}{q(L,\alpha)}$}
	\State{$\cnot{q(L,\alpha)}{q(Z,B_{k}\alpha)}$}
	\EndFor
	\State{}
	\end{algorithmic}
\end{minipage}}
 \\ 
\end{tabular}
\end{center}
\caption{\label{table:automorphisms} Circuits implementing automorphisms of a BB LDPC code within the connectivity already present for syndrome checks. These circuits are fault tolerant and have $\cnotgate$ depth four. If $s = A_j A_k^T$ or $s = B_j B_k^T$, then the logical gate implemented by these automorphisms performs the transformation $\bar X_\alpha, \bar Z_\alpha, \bar X'_\alpha, \bar Z'_\alpha \to \bar X_{s\alpha}, \bar Z_{s\alpha}, \bar X'_{s\alpha}, \bar Z'_{s\alpha}$.}
\end{table}

\subsection{Accessing the Primed Block via a ZX-duality}
\label{ssec:ZX_duality}
 
A ZX-duality is a permutation of the logical qubits that commutes with the stabilizer, except that it turns $X$ checks into $Z$ checks and $Z$ checks into $X$ checks. A physical circuit implementing this permutation and then applying Hadamard to all data qubits always acts as a logical gate \cite{breuckmann2022foldtransversal}. In this section we focus on the implementation of a particular ZX-duality with applications for readout. We leave discovery and implementation of other ZX-dualities for future work. In particular, we derive a general method for constructing fault tolerant circuits for implementing a particular ZX-duality that is present in all BB LDPC codes. While the circuits from this construction are generally quite expensive, they may be amenable to further optimization and can be used sparingly in practice.

Consider a permutation of data qubits that swaps $q(L,\alpha)$ with $q(R,\alpha^{T})$ for all $\alpha \in \mathcal{M}$. A check qubit $q(X,\beta)$ which previously implemented the stabilizer $X(\beta A, \beta B )$ now is connected to the qubits $q(L, (\beta B)^T)$ and $q(R,(\beta A)^T)$ instead, corresponding to the check $Z(\beta^T B^T , \beta^T A^T)$. We see that this permutation switches the stabilizer implemented by $q(X,\beta)$ with the stabilizer implemented by $q(Z,\beta^T)$, so this permutation is indeed a ZX-duality.

We can also see that implementing this permutation and applying Hadamard to all qubits takes logical Pauli matrices to logical Pauli matrices. In particular, the operation swaps $\bar X_\alpha = X(\alpha f,0)$ with $\bar Z'_{\alpha^T} = Z(0,\alpha^T f^T)$, as well as $\bar Z_\alpha = Z(\alpha h^T, \alpha g^T)$ with $\bar X'_{\alpha^T} := X(\alpha^T g, \alpha^T h)$. This operation swaps the primed and unprimed logical blocks, transposes the grid of operators, and applies logical Hadamard to all qubits. Since we can measure logical $X$ for all qubits in the unprimed block using the ancilla system described in \ssec{logical_measurements}, we can use this operation to measure logical $Z$ for qubits in the primed block.

For the rest of this section we describe a fault tolerant method for implementing this operation. We begin with exchanging $q(L)$ and $q(R)$: since these blocks are connected by pairs of edges in $q(X)$ and $q(Z)$, for any $A_i \in A$ and $B_j \in B$ there exists a loop connecting the qubits $q(L,\alpha) \to q(X,A_i^T \alpha) \to q(R, B_j A_i^T \alpha) \to q(Z, B_j \alpha) \to q(L,\alpha)$.
A circuit identical in shape to those in \tab{automorphisms} hence performs a fault tolerant exchange of $q(L,\alpha)$ and $q(R, B_j A_i^T\alpha)$ for all $\alpha$. The additional shift of $B_j A_i^T$ can be removed via an additional automorphism gate. It remains to exchange $q(L,\alpha) \leftrightarrow q(L,\alpha^T)$, as well as $q(R,\alpha)\leftrightarrow q(R,\alpha^T)$ for all $\alpha$, which is significantly more complicated. We focus on $q(L,\alpha) \leftrightarrow q(L,\alpha^T)$ in our discussion but it will be clear the exact same transformations are implementable on $q(R)$ in parallel with those on $q(L)$.

\begin{figure}[t]
  \centering
    \includegraphics[height=8cm]{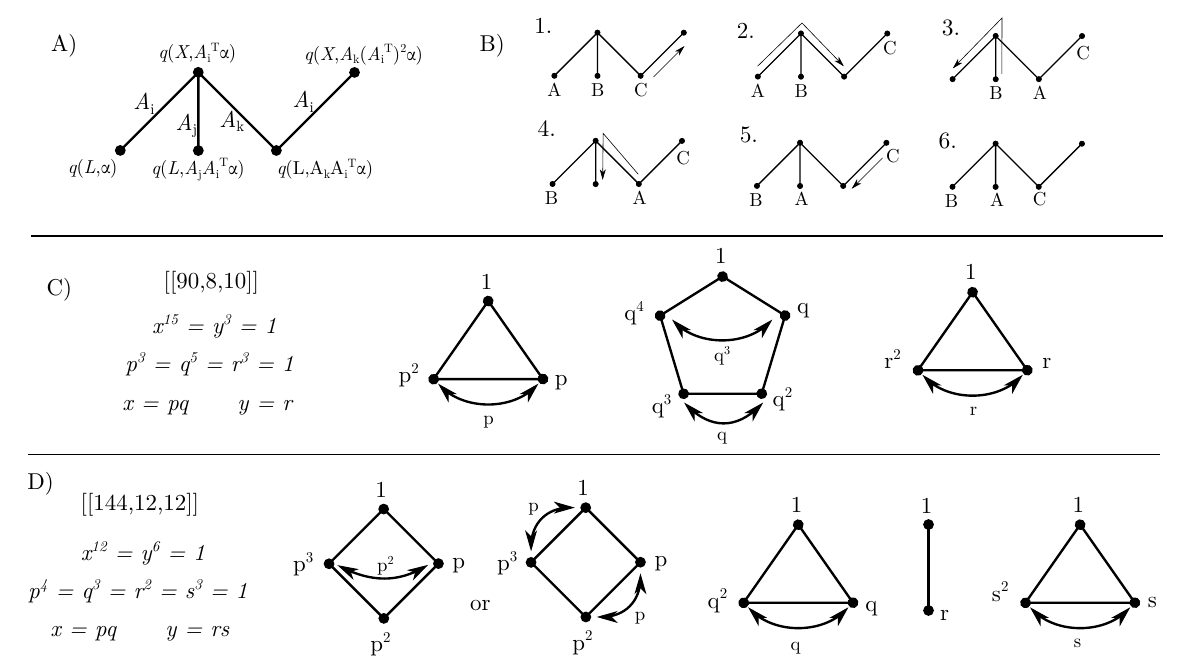}
     \caption{Diagrams for the description of the implementation of the ZX-duality permutation. A) A subgraph of the Tanner graph providing enough connectivity to fault tolerantly swap $q(L,\alpha)$ and $q(L,A_j A_i^T\alpha)$. B) A sequence of shifts of the data on the qubits in the $q(L)$ block that performs the desired exchange without interacting qubits directly. A naive implementation of this sequence has $\cnotgate$ depth 12. C) D) Decomposition of the generators of $\mathcal{M}$ via the classification of finite Abelian groups for two different codes. Drawing the Cayley graph of the subgroup for each generator reveals the ratios defining pairs of qubits that must be exchanged to implement the permutation $q(L,\alpha) \leftrightarrow q(L,\alpha^T)$.}
  \label{fig:duality_implementation}
\end{figure}

Fault tolerant implementation of the permutation $q(L,\alpha) \leftrightarrow q(L,\alpha^T)$ can be achieved using a more sophisticated version of the fault tolerant circuits in \tab{automorphisms} used for implementing automorphisms. These circuits relied on the existence of a connected loop of alternating check and data qubits, enabling a short depth fault tolerant circuit implementing a cyclic permutation of the data qubits therein. The same connectivity can be leveraged to implement a fault tolerant nearest neighbor swap of two data qubits connected by a check qubit. The fault tolerance of these circuits relies on the same principle: while a swap gate acting on two qubits containing data is not fault tolerant, moving a data qubit onto a blank qubit is. \fig{duality_implementation} A) shows a subgraph of the Tanner graph consisting of several connected qubits in the $q(L)$ and $q(X)$ block, and \fig{duality_implementation} B) shows a sequence of operations where two data qubits can be exchanged without ever interacting directly. This gives us the following capability: whenever the circuits in \tab{automorphisms} can implement the cyclic permutation $q(L,\alpha)\to q(L,s\alpha )$ for all $\alpha$, there also exists a circuit that can swap $q(L,\alpha)$ and $q(L,s\alpha)$ for a particular $\alpha$. Matching circuits exist for $q(R)$, and can be implemented simultaneously.

To decompose $q(L,\alpha) \leftrightarrow q(L,\alpha^T)$ into a sequence of swaps, it will be helpful to consider the group structure of $\mathcal{M}$. Consider for example the $[[90,8,10]]$ code with $x^{15}= y^3 = 1$. 
Following the classification of finite Abelian groups we see that $\mathcal{M} \cong \mathbb{Z}_3 \times \mathbb{Z}_5 \times \mathbb{Z}_3$. We can re-express elements of $\mathcal{M}$ using generators $p,q,r$ with $p^3 = q^5 = r^3 =1$ where $x = pq$ and $y = r$. Transforming $\alpha$ to $\alpha^T$ amounts to decomposing $\alpha$ as $\alpha = p^i q^j r^k$ and exchanging the qubit with $\alpha^T = p^{-i} q^{-j} r^{-k}$.

This exchange $p^i q^j r^k \leftrightarrow p^{-i} q^{-j} r^{-k}$ can be split into a sequence of swaps that are implementable with the method described above using \fig{duality_implementation} A) and B). It suffices to be able to exchange for any $\alpha = q^j r^k$ the qubits $q(L,p^i \alpha) \leftrightarrow q(L,\alpha p^{-i})$, as well as for any $\alpha = p^i r^k$ the qubits $q(L,q^j \alpha) \leftrightarrow q(L,\alpha q^{-j})$, and finally for any $\alpha = p^i q^j$ the qubits $q(L,r^k \alpha) \leftrightarrow q(L,\alpha r^{-k})$. This, for any $i,j,k$, creates a sequence of qubits $q(L,p^i q^j r^k) \leftrightarrow q(L,p^{-i} q^{j} r^{k})  \leftrightarrow q(L,p^{-i} q^{-j} r^{k})  \leftrightarrow q(L,p^{-i} q^{-j} r^{-k}) $ where swaps are possible along each nearest neighbor.  This is sufficient for swapping the first and last qubit in the chain. The implementation of the individual generators like $q(L,\alpha q^i) \leftrightarrow q(L,\alpha q^{-i})$ swaps may also involve additional intermediate qubits, but this only lengthens the chain and does not prohibit implementation.

The resources required for swapping $q(L,p^i \alpha) \leftrightarrow q(L,\alpha p^{-i})$ where $p^3 = 1$ and similarly for other generators depends on the order of the generator $p$ as well as the ratios $A_iA_j^T$ that can be formed using terms $A_i,A_j \in  A$ or similar ratios from $B$. See \fig{duality_implementation} C). Plotting the Cayley graph of the cyclic subgroup spanned by $p$ immediately reveals that since $p$ is order three, only a single ratio $B_iB_j^T = p$ is needed in order to swap any qubits marked $p^1$ and $p^2$, while leaving $p^0$ qubits in place. Indeed since $B = 1 + x^2 + x^7 = 1 + p^2 q^2 + p q^2$ we can implement $p = (p^2 q^2)(p q^2)^T$ in a single layer of transforms in \fig{duality_implementation} B). 

The exchange $q(L,\alpha q^i ) \leftrightarrow q(L,\alpha q^{-i})$ with $q^5 = 1$ requires two such ratios $q$ and $q^2$, the minimal depth expression of which demands the chaining together of two such transforms each. In other codes, like the $[[144,12,12]]$ code, we encounter generators $p,q,r,s$ of order $p^4 = q^3 = r^2 = s^3 = 1$. Elements of order two like $r$ require no swaps at all, and elements of order four like $p$ can be implemented either using the ratio $p^2$ or just $p$, as shown in \fig{duality_implementation} D). Numerical searches can quickly compute the most efficient decompositions of the required swaps. We give the orders of the generators, the ratios defining the required swaps, and the number of transforms required to implement them in \tab{duality_implementations}.  

We emphasize that the swap $q(L,p^i \alpha) \leftrightarrow q(L,\alpha p^{-i})$ can be performed for all $\alpha = q^j r^k$ simultaneously in parallel. This stems from the structure of the exchange circuit in \fig{duality_implementation} B). This circuit performs the swap $q(L,\alpha) \leftrightarrow q(L,A_jA_i^T \cdot \alpha)$ while using the qubit $q(L,A_kA_i^T \cdot \alpha)$ as scratch space. However, we can simultaneously want to swap $q(L,A_kA_i^T \cdot \alpha) \leftrightarrow q(L,A_jA_i^T\cdot A_kA_i^T \cdot \alpha)$ since the first step of the exchange circuit in \fig{duality_implementation} B) is to move the data marked `A' away from the qubit holding it, just as if it were a piece of data marked `C' for a different exchange.

For clarity, we compute the total depth of the circuit for the $[[144,12,12]]$ code without any further optimization. The ratios $p,q,s$ can be implemented using two ratios each via $p = x^{-3} y \cdot y^{-2} y$,   $q = y^{-3}x^2 \cdot y^{-3}x^2$ and $s = y^{-2} y \cdot y^{-2} y$. This results in a chain $q(L,\alpha) \leftrightarrow q(L, \alpha')\leftrightarrow q(L, \alpha'') ... \leftrightarrow   q(L,\alpha^T)$ of length six (counting the number of $\leftrightarrow$s). We can swap the qubits at the ends of a chain of length $n$ using $2n{-}1$ many nearest neighbor swaps. Each swap circuit of the form \fig{duality_implementation} B) can be implemented in $\cnotgate$ depth twelve, resulting in $\cnotgate$ depth $(2\cdot 6 -1)\cdot 12 = 132$ to implement $q(L,\alpha  ) \leftrightarrow q(L,\alpha^T)$. 


Despite its fault tolerance, the implementation of this logical operation is clearly significantly more expensive than that of the automorphisms. Since the intermediate permutations corresponding to each of the generators $p,q,r$ are not ZX dualities in general, it will not be possible in general to perform error correction during this long operation. However, the significant overhead of this operation may be worth such a large cost, since it grants us the capability of accessing the primed block of qubits, effectively doubling the storage capacity of the code. This operation can also be used significantly more sparingly than the automorphism gates, and may be amenable to additional optimization. \edit{Alternatively, additional connections and qubits beyond those necessary for the Tanner graph could be introduced to more directly implement the ZX-duality, though it is likely this will sacrifice the thickness-2 property.}

\begin{table}
	\begin{center}
		\begin{tabular}{|c|c|c|c|c|}
			\hline
			Code & Base Order & Reduced Order & Required Ratios & Swap Chain Length \\ \hline \hline
			$[[72,12,6]]$ & $x^6, y^6$ &  $p^2, q^3, r^2, s^3$ & $q,s$ & 4 \\ \hline
			$[[90,8,10]]$ & $x^{15}, y^{3}$ & $p^3, q^5, r^{3}$ & $p, q, q^2, r$ & 6 \\ \hline
			$[[108,8,10]]$ & $x^9, y^6$ &  $p^9, q^2, r^3$ & $p, p^3, p^5, p^7, r$ & 9 \\ \hline
			$[[144,12,12]]$ & $x^{12}, y^6$ & $p^4, q^3, r^2, s^3$  & $p, q, s$ & 6 \\ \hline
			$[[288,12,18]]$ & $x^{12}, y^{12}$ & $p^4, q^{3}, r^4, s^3$ & $p,q,r,s$ & 10 \\ \hline
			$[[360,12,\leq24]]$ & $x^{30}, y^{6}$ & $p^2, q^3, r^5, s^2, t^3$ & $q, ps, psr^2, psr^3, t$ & 11 \\ \hline
		\end{tabular}
	\end{center}
     \caption{\label{table:duality_implementations} Table deriving the steps in the circuit implementing the $q(L,\alpha) \leftrightarrow q(L,\alpha^T)$, permutation for the ZX-duality for several codes. The generators of $\mathcal{M}$ are decomposed into generators following the decomposition of finite Abelian groups. Following \fig{duality_implementation} B) and C) these generators demand a set of ratios of terms in $A$ or $B$ which define fault tolerantly implementable exchanges of qubits with corresponding labels. The result is a decomposition of $q(L,\alpha) \leftrightarrow q(L,\alpha^T)$ into a chain $q(L,\alpha) \leftrightarrow q(L, \alpha')\leftrightarrow q(L, \alpha'') \leftrightarrow ... \leftrightarrow   q(L,\alpha^T)$ with length as shown (counting the number of arrows, $\leftrightarrow$).  Note the special implementation of the ratios in the $[[360,12,\leq24]]$ code: the ratios $r^2,r^3$ are not implementable, but $psr^2, psr^3$ are. This is fine if we can also perform the $ps$ ratio on its own to remove the additional transformation on some of the qubits.}
\end{table}

\subsection{Logical Measurements}
\label{ssec:logical_measurements}

\begin{figure}[t]
  \centering
    \includegraphics[width=0.9\textwidth]{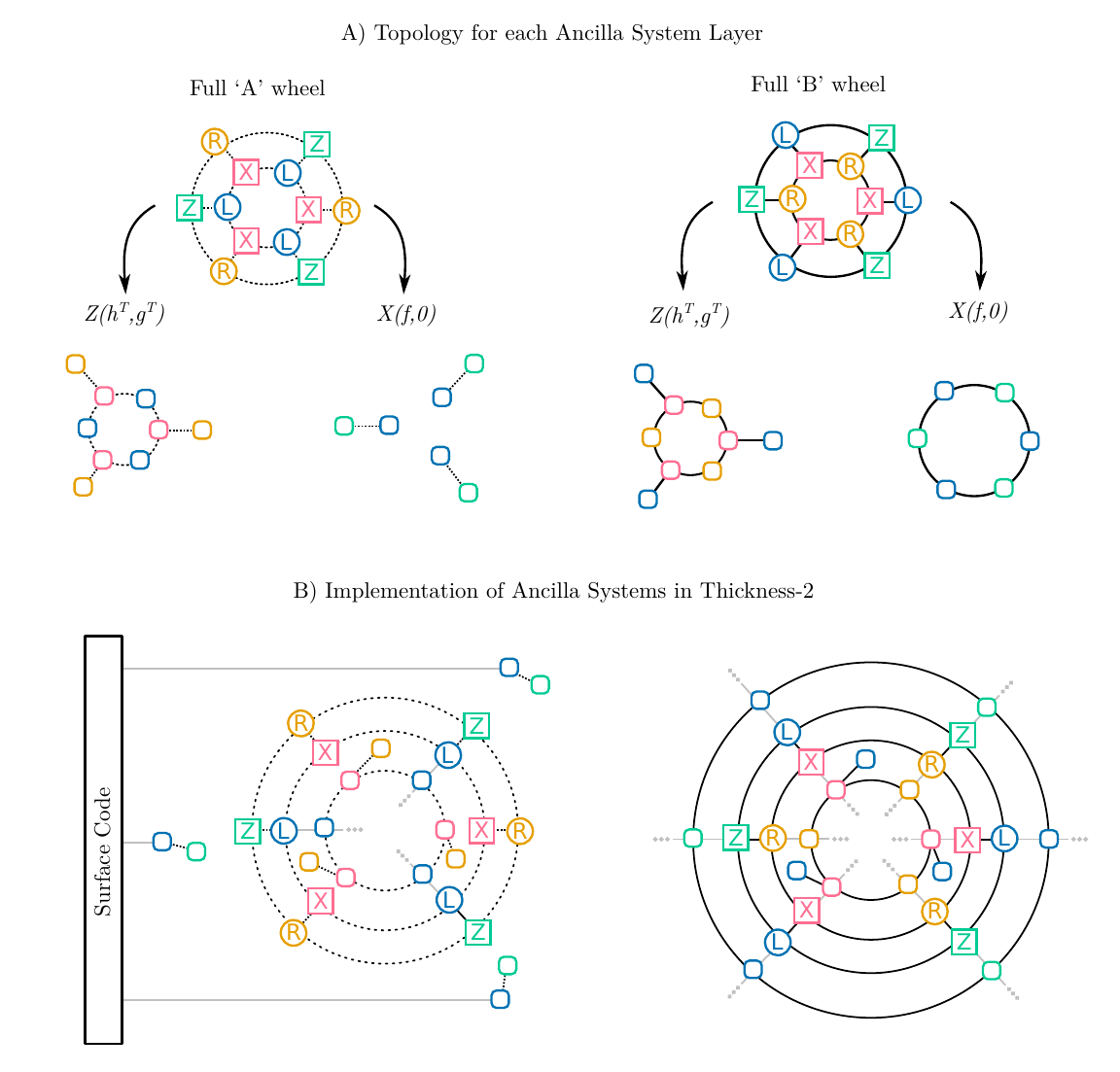}
     \caption{\edit{Illustration} of the thickness-2 property of the Tanner graph of BB LDPC codes. A)  Planar embedding of each layer of the ancilla system from \ifmain{\fig{2Dlayout}~C)}{Figure~\figlayout~C) in the main text} via truncating the `A' and `B' wheels. B) Implementation of the graph from \ifmain{\fig{2Dlayout}~C)}{Figure~\figlayout~C) in the main text} within thickness-2 by nesting wheels. Wheels corresponding to $X(f,0)$ are placed on the outside and wheels corresponding to $Z(h^T,g^T)$ are placed on the inside. The $X(f,0)$ system can be connected to a surface code in the `A' plane. Here we only show one layer per system, but this construction can be repeated for arbitrarily many layers. }
	  
  \label{fig:ancilla_system}
\end{figure}

In this section we describe how to leverage methods from \cite{cohen2022lowoverhead} to implement fault-tolerant measurements of the operators $\bar X_1 = X(f,0)$ and $\bar Z_1 = Z(h^T,g^T)$. As described above, this capability suffices to measure $\bar X$ and $\bar Z$ for all logical qubits. We can also use this technique to measure various Pauli product operators by measuring $\bar X_\alpha, \bar Z_\alpha, \bar X'_\alpha, \bar Z'_\alpha$ for $\alpha$ not corresponding to logical qubits.

The measurement is facilitated by an ancilla system that extends the Tanner graph of the original code. The code defined by this extended Tanner graph contains the logical operator of interest as a stabilizer, enabling its fault tolerant measurement. A sketch of the structure of this ancilla system is given in \ifmain{\fig{2Dlayout}~C)}{Figure~\figlayout~C) in the main text}. For the logical operator $X(f,0)$, we consider a subgraph of the Tanner graph consisting of $q(L,f)$ as well as $q(Z,\alpha)$ operators corresponding to checks with support on $q(L,f)$. Similarly for the logical operator $Z(h^T, g^T)$ we consider a subgraph consisting of $q(L,h^T), q(R,g^T)$ as well as $q(X,\alpha)$ for the relevant $\alpha$. These subgraphs are copied several times and are connected together as shown in the figure: we call the resulting construction an ancilla system.  With enough copies, the code defined by the extended Tanner graph has the same distance as the original code. 

Furthermore, the ancilla system for measuring $X(f,0)$ can be connected to another quantum error correction code, such as a surface code. This enables a joint $\bar X \bar X$ measurement between a surface code qubit and any qubit within the BB code. A subsequent measurement of $Z(h^T,g^T)$ and some additional Pauli corrections then achieves a quantum teleportation circuit.

The main challenge of implementing these ancilla systems, in addition to minimizing their size, is to show that the extended Tanner graph satisfies the thickness-2 constraint. If our goal is to leverage the $X(f,0)$ ancilla system to measure a Pauli product measurement with a surface code qubit, then arguably a thickness-2 extension of the Tanner graph does not suffice since there is no obvious way of connecting it to the surface code qubit as in \ifmain{\fig{2Dlayout}~C)}{Figure~\figlayout~C) in the main text}. To this end, we show how to make the subgraph corresponding to the $X(f,0)$ ancilla system ``effectively planar'': while the graph has thickness-2, the planar graph in one plane consists entirely of connected components with two vertices. Given this property of the embedding of the $X(f,0)$ ancilla system, a connection between this system and a surface code may be facilitated by a construction that is thickness-2 overall.

An effectively planar embedding of the ancilla system relies on the fact that the logical operator $X(f,0)$ has no support on the $q(R)$ block. An implementation of more general logical operators is possible, but would require a graph that renders many ancilla qubits inaccessible from the outside. 

\edit{We briefly give a self-contained description for the construction of the ancilla system from \cite{cohen2022lowoverhead}, following their notation. Suppose we are interested in measuring a logical operator $\bar X$ that is supported on some set of qubits $V_{\bar X}$. Then, let $C_{\bar X}$ be the collection of Pauli-$Z$ checks that have support on any of the $V_{\bar X}$. If we view these as sets of vertices in a Tanner graph, and let $E_{\bar X}$ contain the edges between $V_{\bar X}$ and $C_{\bar X}$, then $\mathcal{G}_{\bar X} := (V_{\bar X},C_{\bar X}, E_{\bar X})$ forms a subgraph of the Tanner graph of the BB code.

The ancilla system is constructed out of copies of `primal layers' isomorphic to $\mathcal{G}_{\bar X}$, and `dual layers' isomorphic to $\mathcal{G}_{\bar X}^{T} := (V^T_{\bar X},C^T_{\bar X}, E^T_{\bar X})$ defined as follows: each $v \in V_{\bar X}$ has a corresponding $v^{T} \in C^{T}_{\bar X}$, each $c \in C_{\bar X}$ has a corresponding $c^{T} \in V^{T}_{\bar X}$, and each edge $(v,c) \in E_{\bar X}$ has a corresponding $(v^T,c^T) \in E^T_{\bar X}$. For some parameter $r$, the final Tanner graph is that of the BB code, plus $r$ additional copies of the dual graph labeled $\mathcal{G}_{\bar X}^{T}[j]$ for $1 \leq j \leq r$, and $r-1$ additional copies of the primal graph labeled $\mathcal{G}_{\bar X}[j]$ for $2 \leq j \leq r$. We regard the $\mathcal{G}_{\bar X}$ within the original code as $\mathcal{G}_{\bar X}[1]$. We also add additional connections between $\mathcal{G}_{\bar X}[j]$ and $\mathcal{G}_{\bar X}^T[j]$ for $j \leq r$, as well as $\mathcal{G}_{\bar X}^T[j]$ and $\mathcal{G}_{\bar X}[j+1]$ for $j < r$: specifically, we connect the associated pairs of $v,v^T$ and $c,c^T$.

It is shown by \cite{cohen2022lowoverhead} that the resulting Tanner graph defines an error correction code of distance $d$ when $r = d$. We construct two such ancilla systems: one for $\bar X := X(f,0)$ and one for $\bar Z := Z(h^T,g^T)$.} \tab{f_polys} shows a choice of $f,g,h$ for the $[[144,12,12]]$ code, defining $X(f,0)$ and $Z(h^T,g^T)$ such that these operators are all minimum weight, and define $\mathcal{G}_{\bar X}$ and $\mathcal{G}_{\bar Z}$ with 30 qubits each. \edit{To achieve $d = 12$ we hence require $2 \times 30 \times (2d-1) = 1380$ additional qubits. We suspect that significantly more efficient variations of this constructions are possible, but leave their development for future work.

The construction presented above is complicated by the fact that vertices in the Tanner graph take on alternating roles in each layer: in the primal layers the vertices $v$ are physical qubits, whereas in the dual layers the $v^T$ are checks. However, for the purposes of giving a thickness-2 decomposition we need not concern ourselves with this. If we do not distinguish between checks and physical qubits, then the primal layers $\mathcal{G}_\mathcal{X}$ and dual layers $\mathcal{G}^T_\mathcal{X}$ have isomorphic Tanner graphs.  Hence, for the purposes of the following, we view all layers as identical.

}

\edit{We now show why a} thickness-2 embedding of the $Z(h^T,g^T)$ ancilla system, and an effectively planar embedding of the $X(f,0)$ ancilla system is possible. This argument is best understood in reference to \fig{ancilla_system}. We begin by understanding the thickness-2 decomposition of each layer of the ancilla systems, leveraging \ifmain{\fig{wheel_extraction}}{Figure~\figwheelextraction in the main text}. In \fig{ancilla_system} A), we can see that \edit{$\mathcal{G}_{\bar Z}$ for} the $Z(h^T,g^T)$ system decomposes into `hairy rings' in both the `A' plane and the `B' plane since it has no support on $q(Z)$. \edit{$\mathcal{G}_{\bar X}$ for} the $X(f,0)$ system is a collection of connected pairs in the `A' plane and collection of rings in the `B' plane, since it has no support on $q(X)$.

\edit{Since $\mathcal{G}_{\bar X},\mathcal{G}_{\bar Z}$ are subgraphs of the BB code's Tanner graph, and its Tanner graph has thickness-2, and since $\mathcal{G}_{\bar X},\mathcal{G}_{\bar Z}$ and $\mathcal{G}^T_{\bar X},\mathcal{G}^T_{\bar Z}$ are isomorphic if we do not distinguish between qubits and checks, we see that each layer of the ancilla construction must be thickness-2 individually. The main challenge is to show that the connections between the layers can be facilitated without introducing any crossings.}

\fig{ancilla_system} B) shows how to connect several layers of the two ancilla systems to both the wheel graphs of the BB code, and also an ancillary surface code. We arrange the wheels of the BB code such that $q(X),q(L)$ are on the inside of the `A' wheels, and that $q(X),q(R)$ are on the inside of the `B' wheels. \edit{$\mathcal{G}_{\bar Z}$ and $\mathcal{G}^T_{\bar Z}$ for} of the $Z(h^T,g^T)$ system can be repeatedly nested inside of the wheels of the BB code. The $q(L)$ qubits can be connected together on the `A' plane, and the $q(R)$ and $q(X)$ qubits can be connected on the `B' plane. As for \edit{$\mathcal{G}_{\bar X}$  and $\mathcal{G}^T_{\bar X}$ for the} $X(f,0)$ system, the rings in the `B' plane can be wrapped around the wheels of the BB code which already allows connection of the required $q(L)$ and $q(Z)$ qubits. This leaves the pairs of connected qubits in the `A' plane completely free of any connections between the layers, making them available to be connected to a surface code system.

\edit{We have considered just two ancilla systems here for measuring $X(f,0)$ and $Z(h^T,g^T)$. However, using additional ancilla systems, especially if their size can be reduced, or equipping these two ancilla systems with additional connections to the $X(g,h)$ and $Z(0,f^T)$ logical operators are potential ways to eliminate the need for the error-prone ZX-duality from \ssec{ZX_duality} and access all logical qubits. On the other hand, it is not clear that either approach would preserve the thickness-2 property.}

\ifdefined\maindocument
\else
    \bibliographystyle{unsrt}
    \bibliography{mybib}
    \end{document}
\fi

\section{Conclusion}
\label{sec:conclusions}

In summary, we offered a new perspective on how a fault-tolerant quantum memory 
could be realized using  near-term quantum processors with a small qubit overhead.
Our approach complements a concatenation-based scheme by 
Pattison,  Krishna, and Preskill
~\cite{pattison2023hierarchical}
where each data qubit of a high-rate  LDPC code is additionally encoded by the surface code.
 Although the  concatenation approach makes use of the high error threshold of the surface code and its geometric locality to address quantum
hardware  limitations such as  a relatively high noise rate and  limited qubit connectivity,
 the additional surface code encoding incurs a significant qubit overhead,
partially negating the advantages offered by LDPC codes. 
Here we have shown that the concatenation step can be avoided by introducing examples of high-rate LDPC codes which have nearly the same error threshold as the surface code itself. 
Although these LDPC codes are not geometrically local, qubit connectivity required for 
syndrome measurements is described by a thickness-two graph which can be implemented
using two planar degree-3 layers of qubit couplers. 
 This is a valid architectural solution for platforms based on superconducting qubits. Numerical simulations performed for the circuit-based noise model
indicate that the proposed LDPC codes compare favorably with the surface code
in the practically relevant range of error rates $p\ge 0.1\%$ 
offering the same level of error suppression with nearly 15$\text{x}$ reduction in the
qubit overhead.

The key hardware challenges to enable the new codes with superconducting qubits are:
\begin{enumerate}
    \item the development of a low-loss second layer,
    \item the development of qubits that can be coupled to 7 connections (6 buses and 1 control line), and
    \item the development of long-range couplers.
\end{enumerate} 

These are all difficult to solve but not impossible. For the first challenge, we can imagine a small change to the packaging \cite{bravyi2022future} which was developed for the IBM Quantum Eagle processor \cite{chow2021}. The simplest would be to place the extra buses on the opposite side of the qubit chip. This would require the development of high Q through substrate vias (TSV) which would be part of the coupling buses and as such would require intensive microwave simulation to make sure these TSVs could support microwave propagation while not introducing large unwanted crosstalk.

The second challenge is an extension of the number of couplers from the heavy hex lattice arrangement \cite{Nation2021} which is 4 (3 couplers and 1 control) to 7. The implication of this is that the cross-resonance gate, which has been the core gate used in large quantum systems for the past few years, would not be the path forward. This is due to the fact that the qubits in the cross-resonance gate are not tunable and as such for a large device with a large number of connections the probability of energy collisions (not just the qubit levels but also higher levels of the transmon) trends to one very fast \cite{Hertzberg2021}. This is because of frequency requirements for the gate to work properly and intrinsic device variability, which is fundamental to Josephson junction fabrication. However, with the tunable coupler \cite{PhysRevApplied.6.064007, PhysRevLett.127.080505}, which was used in the IBM Quantum Egret and is now being developed for the IBM Quantum Heron, this problem no longer exists as the qubits can be designed to be further apart. This new gate is also similar to the gates used by Google Quantum AI \cite{Arute2019}, which have shown that a square lattice arrangement is possible. Extending the coupling map to 7 connections will require significant microwave modeling; however, typical transmons have about 60fF of capacitance and each gate is around 5fF to get the appropriate coupling strengths to the buses, so it is fundamentally possible to develop this coupling map without changing the properties of the transmon qubits which have been shown to have larger coherence and are stable.

The final challenge is the most difficult. For the buses that are short enough so that the fundamental mode can be used the standard circuit QED model holds. However, to demonstrate the 144-qubit code some of the buses will be long enough that we will require frequency engineering. One way to achieve this is with filtering resonators, and a proof of principle experiment was demonstrated in Ref.~\cite{McKay2015}.

Our work leaves several open questions concerning BB LDPC codes and their applications. 
\begin{enumerate}
\item  What are the tradeoffs between the code parameters $n,k,d$ and can one achieve a constant non-zero encoding rate and a growing distance? 
\item  Are there more general LDPC codes compatible with our syndrome measurement circuit(s)? We expect that the same circuit applies to any two-block LDPC code  based on an Abelian group~\cite{lin2023quantum,wang2023abelian}. However, our circuit analysis breaks down 
for non-Abelian groups. 
\item  Our work gives a depth-7 syndrome measurement circuit,  as measured by the number of $\cnotgate$ layers. Is it possible to reduce the circuit depth? Numerical experiments performed for the code $[[72,12,6]]$ indicate that this code may have no depth-6 syndrome measurement (SM) circuit~\cite{AlexanderPrivateComm}. 

\item We observed that a depth-7 SM circuit is not unique.
A natural next step is 
identifying a SM circuit that works best for a particular code.
In addition, it may be possible to improve the circuit-level distance by
using different SM circuits in different syndrome cycles.
Even though some low-weight fault paths are not detectable by any single circuit, such fault paths may be detected if two circuits are used in tandem.

\item How much would the error threshold change for a noise biased towards measurement errors? Note that measurements are  the dominant source of noise for superconducting qubits. Since the considered BB codes have a highly redundant set of check operators, one may expect that they offer extra protection against measurement errors. 
\item The general-purpose BP-OSD decoder used here may not be fast enough to perform error correction in real time. Is there a faster decoder making use of the special structure of BB codes? 
\item How to apply logical gates? While our work gives a fault-tolerant implementation of certain logical gates, these gates offer very limited computational power and are primarily useful for implementing memory capabilities. 
\end{enumerate}

\section*{Acknowledgements}
The authors thank Ben Brown, Oliver Dial, Alexander Ivrii, Tomas Jochym-O'Connor, Yunseong Nam, Matthias Steffen, Kevin Tien, and John Blue for stimulating discussions at various stages of this project.  
The authors acknowledge the IBM Research Cognitive Computing Cluster service for providing resources that have contributed to the research results reported within this paper.

\end{document}